\documentclass[a4paper]{article}
\title{
  The Proximal Surrogate Index:
  Long-Term Treatment Effects 
  under Unobserved Confounding
}
\author{
Ting-Chih Hung\thanks{National Taiwan University. Email: \href{mailto:r12323020@ntu.edu.tw}{r12323020@ntu.edu.tw}.}
\and 
Yu-Chang Chen\thanks{National Taiwan University. Email: \href{mailto:yuchangchen@ntu.edu.tw}{yuchangchen@ntu.edu.tw}.}
}
\date{\today}

\usepackage{settings}

\begin{document}
\maketitle

\begin{abstract}
  We study the identification and estimation of long-term treatment effects under unobserved confounding 
  by combining an experimental sample, where the long-term outcome is missing, 
  with an observational sample, where the treatment assignment is unobserved. 
  While standard surrogate index methods fail 
  when unobserved confounders exist,
  we establish novel identification results by leveraging proxy variables for the unobserved confounders.
  We further develop multiply robust estimation and inference procedures
  based on these results.
  Applying our method to the Job Corps program, 
  we demonstrate its ability to recover experimental benchmarks 
  even when unobserved confounders bias standard surrogate index estimates.
\end{abstract}

\section{Introduction}

\paragraph{Motivation}

Empirical researchers
are often interested in evaluating
the long-term effects
of interventions or treatments.
For example,
labor economists
are interested in evaluating
job training programs
on participants' long-term earnings and employment outcomes,
public economists
are interested in assessing
in-kind transfer programs
on recipients' long-term well-being,
and
health economists
study
medical treatments
on patients' survival and quality of life.

\paragraph{Challenges}

Although randomized controlled trials (RCTs)
are the gold standard
for causal inference,
they often face practical challenges
in measuring long-term outcomes.
These challenges include
limited funding,
high attrition rates,
and lengthy follow-up periods.
As a result,
researchers may only have access
to short-term outcomes
from experimental studies,
while long-term outcomes of interest
are only available
in observational datasets,
such as administrative records or surveys,
where the treatment assignment is unobserved.

\paragraph{Existing Approaches}

The challenge of identifying long-term effects 
from combined experimental and observational data has
been addressed in prior literature.
The most related work is \cite{athey2025surrogate},
who propose a \emph{surrogate index} approach.
Their method leverages short-term surrogate outcomes
to predict long-term outcomes across datasets, 
recovering the long-term treatment effect
under certain identification assumptions.
However,
these assumptions
require
the absence of unobserved confounders
that affect both the treatment assignment
and the long-term outcome,
or the surrogate outcomes
and the long-term outcome.
These assumptions may be violated in many practical settings.
For example, in job training program evaluations,
unobserved ability 
may affect both participation and long-term earnings; 
in healthcare studies, unobserved health consciousness 
may confound the relationship between short-term behaviors 
and long-term health.
These concerns
motivate our work.

\paragraph{Research Question}

In this paper,
we study the nonparametric identification
and estimation
of long-term treatment effects
under unobserved confounding,
by combining experimental and observational data,
allowing for 
unobserved confounding through the use of proxy variables.

\paragraph{Our Contributions}

This paper contributes to the literature in three ways.
First, we study long-term treatment effects in a two-sample surrogate setting 
where the experimental sample lacks the long-term outcome 
and the observational sample lacks treatment status, 
allowing for unobserved confounding.
We establish nonparametric identification by leveraging proxy variables 
and deriving two complementary bridge-based strategies: 
an outcome bridge that imputes the missing long-term outcome in the experimental sample, 
and a surrogate bridge that reweights the observational sample to the experimental target. 
Second, we combine these strategies to obtain a multiply robust identification formula 
that remains valid under four alternative sets of correctly specified nuisance components, 
and we characterize the corresponding efficient influence function (EIF)
and semiparametric local efficiency bound.
Third, we develop cross-fitted estimators based on
double/debiased machine learning (DML)
that accommodate flexible nuisance estimation, 
and we provide consistency, asymptotic normality, and efficiency results.
In an application to Job Corps (JC), our approach closely matches experimental benchmarks in designs 
where standard surrogate index methods are biased.

\paragraph{Related Literature}

Our paper is situated at the intersection of several strands of literature.
First, it is related to the literature on causal inference with surrogate outcomes
\citep{weir2006statistical, vanderweele2013surrogate, joffe2009related}.
These works explore the use of intermediate outcomes
as surrogates for long-term outcomes of interest.
Many criteria have been proposed
to ensure the validity of surrogates \citep{prentice1989surrogate,frangakis2002principal,lauritzen2004discussion}.
However, these criteria may produce ``surrogate paradoxes''
in the presence of unobserved confounding \citep{chen2007criteria}.
Our work addresses this limitation,
offering new strategies to use surrogates.

Second, this paper adds to the growing body of work on
combining different data sources
for causal inference \citep{colnet2024causal},
especially those that leverage short-term outcomes
to estimate long-term treatment effects.
Besides \cite{athey2025surrogate},
\cite{athey2025experimental},
\cite{chenSemiparametricEstimationLongterm2023},
\cite{ghassamiCombiningExperimentalObservational2022b},
and
\cite{imbensLongtermCausalInference2025a},
propose various methods
to combine experimental and observational samples
to estimate long-term effects,
assuming both samples
contain information on the treatment assignment,
which differs from our setting
where the treatment is unobserved
in the observational sample.

Finally, our work contributes to
the recent literature on proximal causal inference
\citep{miaoIdentifyingCausalEffects2018, tchetgentchetgenIntroductionProximalCausal2024, cuiSemiparametricProximalCausal2024},
which
has been applied
to various settings,
including
longitudinal studies
\citep{imbens2021controlling,deaner2023proxycontrolspaneldata,liu2024proximal,ying2023proximal,tchetgentchetgenIntroductionProximalCausal2024},
mediation analysis
\citep{dukes2023proximal,ghassami2025causal},
and optimal treatment regimes
\citep{bai2025proximal,gao2025multiple,shen2023optimal,qi2024proximal}.
But most existing works
focus on single-sample settings,
except for \cite{ghassamiCombiningExperimentalObservational2022b,imbensLongtermCausalInference2025a},
which also consider data combination settings but under different data structures and assumptions.
Our paper provides an extension
by adapting the proximal learning framework to a more challenging
data combination structure, 
where key variables are missing across samples.

\paragraph{Organization of the Paper}

The rest of this paper is organized as follows.
In \cref{sec:setup-and-notation},
we introduce the setup and notation.
\cref{sec:athey-surrogate-index}
reviews the surrogate index approach
proposed by \cite{athey2025surrogate}.
In \cref{sec:identification},
three identification results
are presented.
The estimation and inference procedures
are discussed in \cref{sec:estimation-and-inference}.
\cref{sec:real-data-application}
presents the empirical application.
Finally,
\cref{sec:conclusion} concludes.

\section{Setup and Notation} \label{sec:setup-and-notation}

We consider a setting
where we have access to
an experimental dataset ($\text{E}$) and 
an observational dataset ($\text{O}$),
with sample sizes $N_{\text{E}}$ and $N_{\text{O}}$ respectively. 
It is convenient to view the pooled dataset
as consisting of $N = N_{\text{E}} + N_{\text{O}}$ units,
drawn from a superpopulation that mixes both experimental
and observational distributions,
where each unit belongs to either 
the experimental sample ($G_i = \text{E}$) or 
the observational sample ($G_i = \text{O}$).
We denote the marginal probability of being in the experimental sample
as $\pi_0 \equiv P(G_i = \text{E})$, with $0 < \pi_0 < 1$.

For each unit $i$, there is a binary treatment of interest, 
$A_{i}\in\{0,1\}$, 
a scalar primary outcome, $Y_{i}$, 
a set of intermediate outcomes, $S_{i}$ 
(which we refer to as surrogates), 
and a set of pre-treatment covariates, $X_{i}$. 
The data structure is characterized 
by differential availability across samples. 
In the experimental sample ($\text{E}$), 
the treatment $A_{i}$ is randomly assigned,
and we observe $(A_i, S_i, X_i)$. 
Our goal is to recover the causal effect of $A_{i}$ on $Y_{i}$.
However, the primary long-term outcome $Y_{i}$ 
is not recorded in this sample. 
To provide information on the long-term outcome,
the observational dataset ($\text{O}$) is collected, 
which contains measurements of $(Y_i, S_i, X_i)$. 
Nevertheless,
the treatment status $A_{i}$ is unobserved in the observational sample.

A key challenge in this surrogate setting,
which we will address in \cref{sec:athey-surrogate-index},
arises when the identification assumptions 
required for the existing surrogate index approach
are violated due to unobserved confounders $U_i$ \citep{athey2025surrogate}.
To overcome this, 
we draw insights from recent proximal inference literature
\citep{miaoIdentifyingCausalEffects2018, tchetgentchetgenIntroductionProximalCausal2024, cuiSemiparametricProximalCausal2024}.
We therefore assume that we have access to
two types of proxy variables:
a set of \emph{surrogate-aligned proxies} $Z_{i}$
and a set of \emph{outcome-aligned proxies} $W_{i}$.
While these proxies are not utilized in the 
standard surrogate index approach reviewed in \cref{sec:athey-surrogate-index},
they are essential to our proposed identification strategy 
in \cref{sec:identification}.
The list of variables and their availability across samples
is summarized in \cref{tab:variable-notation}.

\begin{table}[htbp]
  \centering
  \footnotesize
  \caption{Variable notation and availability across samples.}
  \label{tab:variable-notation}
  \begin{tabular}{cccc}
  \toprule
  Variable & Description & Experimental Sample & Observational Sample \\
  \midrule
  $A$ & Treatment indicator & $\checkmark$ & ? \\
  $Y$ & Primary outcome & ? & $\checkmark$ \\
  $S$ & Surrogate outcomes & $\checkmark$ & $\checkmark$ \\
  $X$ & Pre-treatment covariates & $\checkmark$ & $\checkmark$ \\
  $Z$ & Surrogate-aligned proxies & ? & $\checkmark$ \\
  $W$ & Outcome-aligned proxies & $\checkmark$ & $\checkmark$ \\
  $U$ & Unobserved confounders & ? & ? \\
  \bottomrule
  \end{tabular}
  \caption*{\textit{Note:} 
    The check marks ($\checkmark$)
    indicate that the variable is observed
    in the corresponding sample,
    while question marks (?)
    indicate that the variable is unobserved.
  }
\end{table}

To formally define our causal estimand and identification assumptions,
we adopt the potential outcomes framework
\citep{rubin1974estimating, holland1986statistics, imbens2015causal}.
For any variable $V_{i}$ that can be affected 
by the treatment $A_{i}$, 
we denote $V_{i}(a)$ as its potential value 
if $A_i$ were set to $a\in\{0,1\}$. 
Similarly, we denote $V_{i}(s)$ as the potential value 
if the surrogate were set to $s\in\mathcal{S}$, 
and $V_{i}(a,s)$ if both were set. 
We also use the nested potential outcomes notation $V_{i}(a,S_{i}(a'))$
introduced by
\cite{pearl2001direct} and \cite{robins1992identifiability}
to denote the potential value of $V_i$
when the treatment is set to $a$
and the surrogate is set to its potential value $S_i(a')$.

Throughout the paper,
we assume that all random variables
are independent and identically distributed (IID)
samples from a superpopulation,
although some variables are only observed
in one of the two samples.
We use $P(\cdot)$ and $\E[\cdot]$
to denote
the probability and expectation
with respect to this superpopulation,
and use $f(\cdot)$
to denote
the corresponding probability density/mass function
as appropriate.
We drop the subscript $i$
when there is no ambiguity.

\begin{assumption}[Random Sampling]
  \label{asn:random-sampling}
  We observe an IID sample
  \begin{align*}
    \mathcal{D}
    = \bc{D_i}_{i = 1}^N
    = \bc{\mathbf{1}_{\{G_i = \text{O}\}} Y_i, W_i, \mathbf{1}_{\{G_i = \text{O}\}} Z_i, S_i, \mathbf{1}_{\{G_i = \text{E}\}} A_i, X_i, G_i}_{i = 1}^N
  \end{align*}
  from the population
  described in \cref{sec:setup-and-notation},
  where $N = N_{\text{E}} + N_{\text{O}}$
  is the total sample size,
  and $\mathbf{1}_{\{G_i = g\}}$ indicates
  whether unit $i$
  belongs to sample $g \in \{\text{E}, \text{O}\}$.
  We denote $\mathcal{D}_{\text{E}}$
  and $\mathcal{D}_{\text{O}}$
  as the subsets of $\mathcal{D}$
  corresponding to the experimental
  and observational samples, respectively.
\end{assumption}

We further maintain the stable unit treatment value assumption (SUTVA)
\citep{rubin1980discussion}:
for any variable $V$ affected by $A$ and/or $S$,
\begin{align*}
  V = V(A) = V(S) = V(A, S) = V(A, S(A)).
\end{align*}
which links the observed variables
to the potential outcomes.

Our primary estimand of interest
is the average treatment effect (ATE)
on the primary outcome $Y$
in the experimental sample:
\begin{align}
  \label{eq:ate}
  \tau_0 = \E\bs{Y(a) - Y(a') \mid G = \text{E}},
\end{align}
where $a = 1$ and $a' = 0$.

\section{The Surrogate Index and Its Limitations} \label{sec:athey-surrogate-index}

In this section, 
we briefly review the surrogate index approach proposed by \cite{athey2025surrogate} 
and analyze its limitations
in the presence of unobserved confounding.
This discussion motivates our proximal approach.

\subsection{Identification via the Surrogate Index}

\cite{athey2025surrogate} rely on three key assumptions 
to identify the long-term ATE.
First, the standard unconfoundedness assumption
in the experimental sample
($A \ind (Y(a), S(a)) \mid X, G = \text{E}$)
is assumed to hold.
Second, the \emph{surrogacy} assumption ($A \ind Y \mid S, X, G = \text{E}$),
originally proposed by \cite{prentice1989surrogate},
posits that 
the treatment $A$
provides no additional information about the primary outcome $Y$
beyond what is already contained in the surrogates $S$
and pre-treatment variables $X$
in the experimental sample.
Third, the \emph{comparability} assumption
requires that the conditional outcome distribution 
is stable across samples, i.e., $G \perp Y \mid S, X$.

Under these assumptions, the mean potential outcome is identified by
\begin{align}
  \E\bs{Y(a) \mid G = \text{E}} =
  \E\bs{\E\bs{\mu(S, X, \text{O}) \mid A = a, X, G = \text{E}} \mid G = \text{E}},
  \label{eq:surrogate-index-identification-intro}
\end{align}
where $\mu(S, X, g) \equiv \E[Y \mid S, X, G = g]$
is called the \emph{surrogate index}.

\subsection{Limitations and Motivation}

While elegant, 
the surrogate index approach
relies heavily on the validity of
the surrogacy and comparability assumptions.
\cref{fig:surrogate-model-violations}
illustrates three common scenarios
where these assumptions can be violated.

The first scenario (\cref{fig:surrogate-model-treatment-outcome-confounder}) 
involves an unobserved confounder $U$
that affects both the treatment $A$
and the primary outcome $Y$.
Structurally, this resembles the front-door model \citep{pearl1995causal}.
However, in our context, conditioning on $S$ opens the collider path $G \to A \leftarrow U \to Y$, 
violating the comparability assumption.

The second scenario (\cref{fig:surrogate-model-surrogate-outcome-confounder}) 
involves an unobserved confounder $U$
that affects both the surrogates $S$
and the primary outcome $Y$,
resembling an instrumental variable (IV) model \citep{wrightTariffAnimalVegetable1928,imbensIdentificationEstimationLocal1994,imbensInstrumentalVariablesEconometricians2014}.
Here, 
since $S$ is a collider,
conditioning on $S$ opens the collider path $A \to S \leftarrow U \to Y$,
violating the surrogacy and potentially leading to surrogate paradoxes \citep{chen2007criteria,frangakis2002principal}.
In addition,
the sample indicator $G$
becomes associated with $Y$
after conditioning on $S$,
violating the comparability assumption.

The third scenario (\cref{fig:surrogate-model-direct-effect})
involves a direct effect of $A$ on $Y$ that bypasses $S$,
as in mediation analysis \citep{pearl2001direct, robins1992identifiability, vanderweele2015explanation}.
Since $A$ and $Y$ are never observed together, this effect is fundamentally unidentified.
Like the standard surrogate index,
our proximal approach must assume the absence of direct effects (\cref{asn:no-direct-effect}).

We summarize the key features of these three scenarios
in comparison to the surrogate index setting
in \cref{tab:comparison-of-scenarios}.

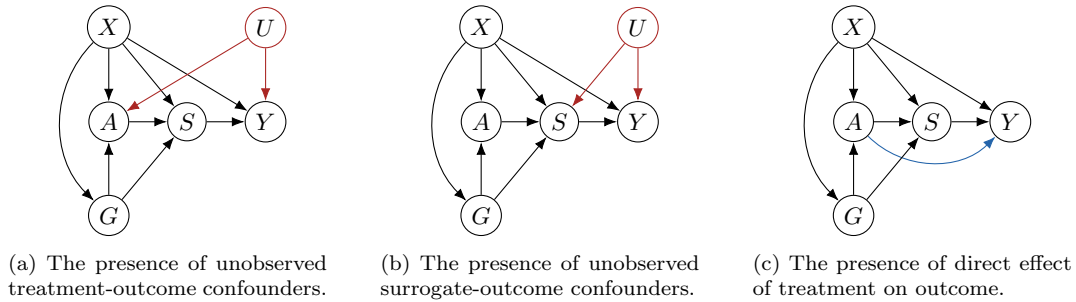
\begin{figure}[htbp]
  \begin{subfigure}[t]{0.3\textwidth}
    \centering
    \small
    \begin{tikzpicture}[scale=1.1]
      \node[name=a, cir]{$A$};
      \node[name=s, right=5mm of a, cir]{$S$};
      \node[name=y, right=5mm of s, cir]{$Y$};
      \node[name=g, below=7mm of a, cir]{$G$};
      \node[name=x, above=7mm of a, cir]{$X$};
      \node[name=u, above=7mm of y, cir, draw=darkred]{$U$};
      \draw[-Latex](a) to (s);
      \draw[-Latex](s) to (y);
      \draw[-Latex](g) to (a);
      \draw[-Latex](g) to (s);
      \draw[-Latex](x) to (a);
      \draw[-Latex](x) to (s);
      \draw[-Latex](x) to (y);
      \draw[-Latex, bend right=45](x) to (g);
      \draw[-Latex, darkred](u) to (a);
      \draw[-Latex, darkred](u) to (y);
    \end{tikzpicture}
    \caption{The presence of unobserved treatment-outcome confounders.}
    \label{fig:surrogate-model-treatment-outcome-confounder}
  \end{subfigure}
  \hfill
  \begin{subfigure}[t]{0.3\textwidth}
    \centering
    \small
    \begin{tikzpicture}[scale=1.1]
      \node[name=a, cir]{$A$};
      \node[name=s, right=5mm of a, cir]{$S$};
      \node[name=y, right=5mm of s, cir]{$Y$};
      \node[name=g, below=7mm of a, cir]{$G$};
      \node[name=x, above=7mm of a, cir]{$X$};
      \node[name=u, above=7mm of y, cir, draw=darkred]{$U$};
      \draw[-Latex](a) to (s);
      \draw[-Latex](s) to (y);
      \draw[-Latex](g) to (a);
      \draw[-Latex](g) to (s);
      \draw[-Latex](x) to (a);
      \draw[-Latex](x) to (s);
      \draw[-Latex](x) to (y);
      \draw[-Latex, bend right=45](x) to (g);
      \draw[-Latex, darkred](u) to (s);
      \draw[-Latex, darkred](u) to (y);
    \end{tikzpicture}
    \caption{The presence of unobserved surrogate-outcome confounders.}
    \label{fig:surrogate-model-surrogate-outcome-confounder}
  \end{subfigure}
  \hfill
  \begin{subfigure}[t]{0.3\textwidth}
    \centering
    \small
    \begin{tikzpicture}[scale=1.1]
      \node[name=a, cir]{$A$};
      \node[name=s, right=5mm of a, cir]{$S$};
      \node[name=y, right=5mm of s, cir]{$Y$};
      \node[name=g, below=7mm of a, cir]{$G$};
      \node[name=x, above=7mm of a, cir]{$X$};

      \draw[-Latex](a) to (s);
      \draw[-Latex, bend right=45, greenblue](a) to (y);
      \draw[-Latex](s) to (y);
      \draw[-Latex](g) to (a);
      \draw[-Latex](g) to (s);
      \draw[-Latex](x) to (a);
      \draw[-Latex](x) to (s);
      \draw[-Latex](x) to (y);
      \draw[-Latex, bend right=45](x) to (g);
    \end{tikzpicture}
    \caption{The presence of direct effect of treatment on outcome.}
    \label{fig:surrogate-model-direct-effect}
  \end{subfigure}
  \caption{Possible violations of the surrogacy or comparability assumptions.}
  \label{fig:surrogate-model-violations}
\end{figure}

These limitations point to a clear strategy inspired by the front-door criterion.
Recall that the front-door formula identifies the causal effect $A \to Y$ 
by composing the effect of $A$ the mediator $S$, and the effect of $S$ on $Y$.
In our setting, the first component ($A \to S$) is readily identified from the experimental sample due to randomization.
The challenge lies entirely in identifying the second component ($S \to Y$) from the observational sample, 
where the relationship is confounded by $U$.
While the standard front-door approach assumes no unobserved $S$-$Y$ confounding, 
we relax this by leveraging proxy variable ($W, Z$) to adjust for $U$.
Thus, our approach conceptually mirrors the front-door strategy 
but adapts it to the data combination setting: 
we combine the experimentally identified effect of $A$ on $S$ 
with the proximally identified effect of $S$ on $Y$ to recover the long-term ATE.

\begin{table}[htbp]
  \centering
  \footnotesize
  \caption{Surrogacy, front-door, IV, and mediation.}
  \label{tab:comparison-of-scenarios}
  \begin{tabular}{lcccc}
  \toprule
   & Surrogacy & Front-door & IV & Mediation \\
  \midrule
  Role of $A$ & Treatment & Treatment & Instrument & Treatment \\
  Role of $S$ & Surrogate & Mediator & Treatment & Mediator \\
  Role of $Y$ & Outcome & Outcome & Outcome & Outcome \\
  Parameter of Interest & $A \to Y$ & $A \to Y$ & $S \to Y$ & $A \to Y$ \\
                        & & & & $A \to S \to Y$ \\
  \midrule
  Feasible w/o joint $(A, Y)$ & Yes & No & No & No \\
  Unobs.\ $A$-$Y$ confounders & No & Yes & No & No \\
  Unobs.\ $S$-$Y$ confounders & No & No & Yes & No \\
  Direct effect of $A$ on $Y$ & No & No & No & Yes \\
  \bottomrule
  \end{tabular}
  \caption*{\textit{Note:} 
    We use ``surrogacy'' to refer to
    \cite{athey2025surrogate}'s surrogate index setting.
    $A \to Y$ and $S \to Y$ denote
    the causal effect of $A$ on $Y$
    and that of $S$ on $Y$ respectively.
    $A \to S \to Y$ denotes
    the indirect effect of $A$ on $Y$
    mediated through $S$.
  }
\end{table}

\section{Restoring Identification via Proxies} \label{sec:identification}

As discussed in \cref{sec:athey-surrogate-index}, 
identifying the long-term ATE requires addressing the
potential violations of surrogacy and comparability assumptions.
In this section, we present a proximal identification strategy that allows for unobserved confounders $U$.
To provide some intuition,
we adopt the single-world intervention graph (SWIG)
to visualize the causal structure
under a hypothetical intervention \citep{richardson2013single, heckman2015causal}.

We first maintain the following no direct effect assumption.
Similar assumptions
have been widely used
in the IV literature,
where it is often referred to as the exclusion restriction
\citep{imbensIdentificationEstimationLocal1994,imbensInstrumentalVariablesEconometricians2014}.
This assumption requires that $A$ has no direct causal link to $Y$ other than through $S$, 
corresponding to the absence of an edge $a \to Y(a,s)$ in \cref{fig:proximal-surrogate-model-swig}.

\begin{assumption}[No Direct Effect]
  \label{asn:no-direct-effect}
  \begin{align}
    Y(a, s) = Y(a', s) \quad \text{for all } a, a' \in \{0, 1\}, s \in \mathcal{S}.
  \end{align}
\end{assumption}

Next, we characterize the confounding structure.
While unobserved confounders $U$ may bias the relationships in the observational sample, 
we assume that $U$, 
together with
the observed covariates $X$, 
are the only sources of confounding.
In terms of \cref{fig:proximal-surrogate-model-swig}, 
this means that there are no other unmeasured common causes 
between $A$ and $Y(a, s)$,
and between $S$ and $Y(s)$,
besides $U$.

\begin{assumption}[Unconfoundedness in the Observational Sample]
  \label{asn:unconfoundedness-observational}
  For all $a \in \{0, 1\}$ and $s \in \mathcal{S}$,
  \begin{align}
    A \ind \bp{\bp{Y(a)}_{a \in \{0, 1\}}, \bp{S(a)}_{a \in \{0, 1\}}} &\mid (X, U, G = \text{O}), \label{eq:unconfoundedness-observational-a} \\
    S \ind \bp{Y(s)}_{s \in \mathcal{S}} &\mid (A, X, U, G = \text{O}), \label{eq:unconfoundedness-observational-s}
  \end{align}
  where $0 < P(A = 1 \mid X, U, G = \text{O}) < 1$ almost surely.
\end{assumption}

In addition,
we assume that
treatment assignment
is randomized
in the experimental sample,
as is standard in RCTs.
However, the surrogate $S$ is an intermediate outcome, not a randomized intervention.
Even within the experiment, individuals with the same treatment status 
may have different surrogate values due to unobserved factors $U$ (e.g., motivation or ability), 
which also affect the long-term outcome $Y$.
Similarly to \cref{asn:unconfoundedness-observational},
we assume that $U$ is the only source of unobserved confounding
in the experimental sample as well.
The following assumption formalizes these conditions.

\begin{assumption}[Unconfoundedness in the Experimental Sample]
  \label{asn:unconfoundedness-experimental}
  For all $a \in \{0, 1\}$ and $s \in \mathcal{S}$,
  \begin{align}
    A \ind \bp{\bp{Y(a)}_{a \in \{0, 1\}}, \bp{S(a)}_{a \in \{0, 1\}}, U, W} &\mid (X, G = \text{E}), \label{eq:unconfoundedness-experimental-a} \\
    S \ind \bp{Y(s)}_{s \in \mathcal{S}} &\mid (A, X, U, G = \text{E}), \label{eq:unconfoundedness-experimental-s}
  \end{align}
  where $0 < e_0(X) < 1$ almost surely,
  and $e_0(X) \equiv P(A = 1 \mid X, G = \text{E})$.
\end{assumption}

\begin{remark}
  The SWIG in
  \cref{fig:proximal-surrogate-model-swig},
  is a pooled representation. 
  The randomization of $A$
  in the experimental sample
  is captured
  by an directed edge from $G$ to $A$,
  indicating that the treatment assignment mechanism
  depends on the sample indicator $G$.
  Although the pooled graph depicts a generic dependence $U \to A$, 
  the experimental design explicitly breaks this dependence
  within the $G=\text{E}$ stratum.
\end{remark}

To address the unobserved confounding,
we need to leverage
proxy variables
that provide indirect information
about the unobserved confounders $U$.
We make the following assumption
regarding the availability of such proxies.

\begin{assumption}[Proxies Availability]
  \label{asn:availability-of-proxies}
  We observe $W$ in both samples and $Z$ only in the observational sample.
  Moreover,
  \begin{align}
    Z \ind Y &\mid (S, X, U, G), \\
    W \ind (A, S, Z) &\mid (X, U, G).
  \end{align}
\end{assumption}

\cref{fig:proximal-surrogate-model-swig} 
illustrates the roles of these proxies.
In this representation, $Z$ acts as a surrogate-aligned proxy.
While the figure depicts $Z$ as a post-treatment variable (e.g., a post-program survey), 
the formal requirement is broader:
$Z$ must be conditionally independent of $Y$ given $S$, $X$, and $U$,
regardless of its temporal or causal positioning.
Conversely, $W$ acts as an outcome-aligned proxy.
Graphically, it behaves like a pre-treatment covariate that is a noisy measure of $U$.
The key requirement is that $W$ is sufficiently informative about $U$ 
but is structurally independent of the $(A, S, Z)$ given covariates.

The proxies satisfying \cref{asn:availability-of-proxies}
are ubiquitous in practical applications.
For example,
in job training programs,
test scores collected before the training
can serve as outcome-aligned proxies $W$,
because they reflect
the participants' underlying abilities $U$,
but do not affect or are affected by
the training $A$ or the short-term outcomes $S$ directly.
Similarly,
post-training surveys about job search self-efficacy 
or soft skills can serve as surrogate-aligned proxies $Z$. 
These measures capture the participants' latent motivation $U$ 
and the immediate psychological impact of the training $A$, 
but arguably do not directly determine long-term earnings $Y$ 
once the actual intermediate employment history $S$ is accounted for.

To generalize the relationships learned in the observational sample 
to the experimental sample, 
we require the mechanism underlying the data generation to be comparable across groups.
In the proximal learning framework, 
this transportability must hold for 
both the primary outcome and the outcome-aligned proxy.

\begin{assumption}[Transportability]
  \label{asn:transportability}
  \begin{align}
    Y &\ind G \mid (S, X, U), \label{eq:transportability-outcome} \\
    W &\ind G \mid (X, U), \label{eq:transportability-proxy}
  \end{align}
  and
  \begin{align}
    \frac{P(G = \text{E} \mid U, X)}{P(G = \text{O} \mid U, X)} < \infty \quad \text{almost surely}.
  \end{align}
\end{assumption}

\cref{eq:transportability-outcome}
is analogous to
the comparability assumption
in \cite{athey2025surrogate},
but allows for unobserved confounders $U$.
It posits that given the surrogate 
and the complete set of confounders (both observed and unobserved),
the distribution of the primary outcome $Y$
is stable across samples.
\cref{eq:transportability-proxy}
is a new requirement specific to the proximal approach.
It ensures that the outcome-aligned proxy $W$ 
relates to the latent confounder $U$ in an invariant manner across the
experimental and observational samples.
In our SWIG representation (\cref{fig:proximal-surrogate-model-swig}), 
these assumptions correspond to the requirement that the group indicator $G$ 
has no directed edges on $Y$ or $W$.
These transportability assumptions are closely related to 
the concept of S-admissibility in the data fusion literature
\citep{bareinboim2016causal, pearl2014external}.

Notably, our identification results do not require transportability for the surrogate-aligned proxy $Z$.
While $Z$ is crucial for identifying the bridge functions within the observational sample,
its distributional relationship with $U$
may vary across samples.
Thus, we allow for arbitrary dependence between $G$ and $Z$, 
accommodating scenarios where the proxy measurement mechanism
for $Z$ differs across samples.
In the SWIG (\cref{fig:proximal-surrogate-model-swig}),
this is reflected by the presence of a directed edge from $G$ to $Z$.

\begin{figure}[htbp]
  \centering
  \small
  \begin{tikzpicture}[scale=1.6]
    \node[name=a, shape=swig vsplit]{
      \nodepart{left}{$A$}
      \nodepart{right}{$a$}
    };
    \node[name=s, shape=swig vsplit] at (2, 0) {
      \nodepart{left}{$S(a)$}
      \nodepart{right}{$s$}
    };
    \node[name=y, ell] at (4, 0) {$Y(a, s)$};
    \node[name=g, cir] at (1, -1) {$G$};
    \node[name=x, cir] at (3, -1) {$X$};
    \node[name=u, cir, draw=darkred] at (3, 1) {$U$};
    \node[name=z, ell] at (1, 1) {$Z(a)$};
    \node[name=w, cir] at (5, 1) {$W$};

    \draw[-Latex](a) to (s);
    \draw[-Latex, in=180, out=110](a) to (z);
    \draw[-Latex](s) to (y);
    \draw[-Latex, in=250, out=180](g) to (a);
    \draw[-Latex](g) to (s);
    \draw[-Latex, in=240, out=120](g) to (z);
    \draw[-Latex, in=270, out=150, looseness=0.5](x) to (a);
    \draw[-Latex, in=270, out=150](x) to (s);
    \draw[-Latex](x) to (y);
    \draw[-Latex](x) to (g);
    \draw[-Latex, bend left=45](x) to (z);
    \draw[-Latex, bend right=45](x) to (w);
    \draw[-Latex, in=90, out=210, darkred](u) to (s);
    \draw[-Latex, darkred](u) to (z);
    \draw[-Latex, darkred](u) to (y);
    \draw[-Latex, darkred, in=90, out=210, looseness=0.5](u) to (a);
    \draw[-Latex, darkred](u) to (x);
    \draw[-Latex, darkred](u) to (w);
    \draw[-Latex, darkred, in=90, out=180, looseness=1](u) to (g);
    \draw[-Latex](z) to (s);
    \draw[-Latex](w) to (y);
  \end{tikzpicture}
  \caption{The SWIG for the underlying causal model
    in the proximal surrogate approach
    when intervening on $A$ and $S$.}
  \label{fig:proximal-surrogate-model-swig}
\end{figure}
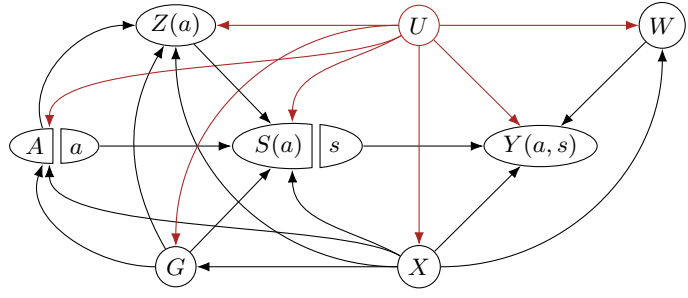

\subsection{Identification via Outcome Bridge Function} \label{sec:identification-via-outcome-bridge-function}

We first introduce
the outcome bridge function,
which helps to impute
the missing long-term outcome $Y$
in the experimental sample.

\begin{assumption}[Outcome Bridge Function]
  \label{asn:outcome-bridge-function}
  There exists a square-integrable function 
  $h_0: \mathcal{W} \times \mathcal{S} \times \mathcal{X} \to \mathbb{R}$
  for the observational sample $(G = \text{O})$
  such that
  \begin{align}
    \E\bs{Y \mid Z, S, X, G = \text{O}}
    = \E\bs{h_0 (W, S, X) \mid Z, S, X, G = \text{O}}.
    \label{eq:outcome-bridge-function}
  \end{align}
\end{assumption}

We further assume
that the surrogate-aligned proxy $Z$
is sufficiently informative
about the unobserved confounders $U$
to ensure that we can learn the unobserved confounding relationship
using $Z$.

\begin{assumption}[Completeness for Outcome Bridge Function]
  \label{asn:completeness-outcome}
  For any $g \in L_2$,
  if $\E\bs{g(U) \mid Z = z, S = s, X = x} = 0$ for all $z, s, x$,
  then $g(U) = 0$ almost surely.
\end{assumption}

This completeness condition
has been widely used
in the nonparametric IV models
\citep{neweyInstrumentalVariableEstimation2003,ai2003efficient,chernozhukov2005iv},
measurement error models \citep{hu2008instrumental,an2012well},
and the recent proximal inference literature
\citep{miaoIdentifyingCausalEffects2018,tchetgentchetgenIntroductionProximalCausal2024,cuiSemiparametricProximalCausal2024,miaoConfoundingBridgeApproach2024,ghassamiCombiningExperimentalObservational2022b,imbensLongtermCausalInference2025a}.

We establish the first identification result
using the outcome bridge function.

\begin{theorem}[Identification via Outcome Bridge Function]
  \label{thm:outcome-bridge-identification}
  \begin{enumerate}
    \item[]
    \item Under \cref{asn:outcome-bridge-function,asn:availability-of-proxies,asn:completeness-outcome},
      the function $h_0$ in \cref{eq:outcome-bridge-function}
      satisfies
      \begin{align*}
        \E\bs{Y \mid U, S, X, G = \text{O}}
        = \E\bs{h_0(W, S, X) \mid U, X, G = \text{O}}.
      \end{align*}

    \item Suppose that the function $h_0$ satisfies \cref{eq:outcome-bridge-function}.
      Under \cref{asn:unconfoundedness-observational,asn:unconfoundedness-experimental,asn:transportability},
      the function $h_0$
      also satisfies
      \begin{align*}
        \E\bs{Y \mid U, S, X, G = \text{E}}
        &= \E\bs{h_0(W, S, X) \mid U, X, G = \text{E}}.
      \end{align*}
      That is,
      the same function $h_0$ transports from the observational sample
      to the experimental sample.

    \item Under \cref{asn:unconfoundedness-observational,asn:unconfoundedness-experimental,asn:no-direct-effect,asn:availability-of-proxies,asn:outcome-bridge-function,asn:completeness-outcome,asn:transportability},
      the mean potential outcome is identified by
      \begin{align}
        \E\bs{Y(a) \mid X, G = \text{E}}
        = \E\bs{h_0(W, S, X) \mid A = a, X, G = \text{E}},
        \label{eq:outcome-bridge-identification}
      \end{align}
      where $h_0$ satisfies \cref{eq:outcome-bridge-function},
      and we call $h_0$ the \emph{proximal surrogate index}.
      Thus, the ATE in the experimental sample is identified by
      \begin{align*}
        \tau_0
        &= \E\bs{\bar{h}_0(1, X) - \bar{h}_0(0, X) \mid G = \text{E}} \\
        &= \E\bs{\frac{A h_0(W, S, X)}{e_0(X)} - \frac{(1 - A) h_0(W, S, X)}{1 - e_0(X)} \;\middle|\; G = \text{E}},
      \end{align*}
      where for any $h: \mathcal{W} \times \mathcal{S} \times \mathcal{X} \to \mathbb{R}$,
      we define
      \begin{align*}
        \bar{h}^{(h)}(A, X)
        \equiv \E\bs{h(W, S, X) \mid A, X, G = \text{E}}.
      \end{align*}
      In particular, $\bar{h}_0 \equiv \bar{h}^{(h_0)}$.
  \end{enumerate}
\end{theorem}

\begin{proof}
  See \cref{sec:proof-of-thm-outcome-bridge-identification}.
\end{proof}

\cref{thm:outcome-bridge-identification}
shows that
the missing long-term outcome $Y$
can be imputed
in the experimental sample
using the outcome bridge function $h$
learned from the observational sample.
The long-term ATE in the experimental sample
can then be identified
using standard methods,
such as outcome regression (OR)
or inverse probability weighting (IPW),
based on the imputed outcomes.

\begin{remark}[Connection to \cite{athey2025surrogate}]
  \cref{thm:outcome-bridge-identification}
  generalizes \cite{athey2025surrogate}'s surrogate index identification
  by allowing for unobserved confounding
  between treatment and outcome
  in the observational sample.
  See \cref{sec:connection-to-surrogate-index-identification}
  for more details.
\end{remark}

\subsection{Identification via Surrogate Bridge Function} \label{sec:identification-via-surrogate-bridge-function}

An alternative identification strategy
leverages
the surrogate bridge function,
as introduced below.

\begin{assumption}[Surrogate Bridge Function]
  \label{asn:surrogate-bridge-function}
  For each treatment level $a \in \{0, 1\}$,
  there exists a square-integrable function 
  $q_{a, 0}: \mathcal{Z} \times \mathcal{S} \times \mathcal{X} \to \mathbb{R}$
  such that
  \begin{align}
    \E\bs{q_{a, 0}(Z, S, X) \mid W, S, X, G = \text{O}}
    = \frac{f(W, S \mid A = a, X, G = \text{E}) f(X \mid G = \text{E})}{f(W, S, X \mid G = \text{O})}.
    \label{eq:surrogate-bridge-function}
  \end{align}
\end{assumption}

To ensure that 
the surrogate bridge function (\cref{eq:surrogate-bridge-function})
can be used for reweighting,
we impose the following completeness condition.

\begin{assumption}[Completeness for Surrogate Bridge Function]
  \label{asn:completeness-surrogate}
  For any $g \in L_2$,
  if $\E\bs{g(U) \mid W = w, S = s, X = x} = 0$ for all $w, s, x$,
  then $g(U) = 0$ almost surely.
\end{assumption}

This completeness condition mirrors
that in \cref{asn:completeness-outcome},
guaranteeing
the outcome-aligned proxy $W$
is sufficiently informative
about the unobserved confounders $U$.

We now present
an alternative identification result
based on the surrogate bridge function.

\begin{theorem}[Identification via Surrogate Bridge Function]
  \label{thm:surrogate-bridge-identification}
  \begin{enumerate}
    \item[]
    \item Under \cref{asn:surrogate-bridge-function,asn:availability-of-proxies,asn:completeness-surrogate},
      the function $q_a$ in \cref{eq:surrogate-bridge-function}
      satisfies
      \begin{align*}
        \E\bs{q_{a, 0}(Z, S, X) \mid U, S, X, G = \text{O}}
        = \frac{f(U, S \mid A = a, X, G = \text{E}) f(X \mid G = \text{E})}{f(U, S, X \mid G = \text{O})}.
      \end{align*}

    \item Under \cref{asn:unconfoundedness-observational,asn:unconfoundedness-experimental,asn:no-direct-effect,asn:availability-of-proxies,asn:surrogate-bridge-function,asn:completeness-surrogate,asn:transportability},
      the mean potential outcome is identified by
      \begin{align}
        \E\bs{Y(a) \mid G = \text{E}}
        = \E\bs{q_{a, 0}(Z, S, X) Y \mid G = \text{O}},
        \label{eq:surrogate-bridge-identification}
      \end{align}
      where $q_{a, 0}$ satisfies \cref{eq:surrogate-bridge-function}.
      Thus, the ATE in the experimental sample is identified
      by
      \begin{align*}
        \tau_0 = \E\bs{q_{1, 0}(Z, S, X) Y \mid G = \text{O}} - \E\bs{q_{0, 0}(Z, S, X) Y \mid G = \text{O}}.
      \end{align*}
  \end{enumerate}
\end{theorem}

\begin{proof}
  See \cref{sec:proof-of-thm-surrogate-bridge-identification}.
\end{proof}

\cref{thm:surrogate-bridge-identification}
shows that
the long-term ATE
in the experimental sample
can be identified
by reweighting
the observed outcomes $Y$
in the observational sample
using the surrogate bridge functions $q_{a, 0}$'s.

\begin{remark}
  \cref{thm:surrogate-bridge-identification}
  is new to the literature,
  providing a novel form of bridge function (\cref{eq:surrogate-bridge-function})
  for identifying the ATE
  in the data combination setting.
\end{remark}

\begin{remark}[Connection to \cite{athey2025surrogate}]
  \cref{thm:surrogate-bridge-identification}
  generalizes
  \cite{athey2025surrogate}'s surrogate score identification
  by allowing for unobserved confounding.
  See \cref{sec:connection-to-surrogate-score-identification}
  for more details.
\end{remark}

\subsection{Multiply Robust Identification}

In 
\cref{sec:identification-via-outcome-bridge-function,sec:identification-via-surrogate-bridge-function},
we present two distinct identification strategies
based on
the outcome bridge functions
and the surrogate bridge functions,
respectively.
Combining these two strategies,
we obtain
a multiply robust identification result
as follows.

\begin{theorem}[Multiply Robust Identification]
  \label{thm:mutiply-robust-identification}
  Let $\tau_0$ be the true ATE defined in \cref{eq:ate}.
  Let $\eta = (e, h, \bar{h}, q_0, q_1)$
  be a set of possibly misspecified nuisance functions.
  Let $\eta_0 = (e_0, h_0, \bar{h}_0, q_{0,0}, q_{1,0})$
  denote the true nuisance functions,
  where $e_0$ is the true propensity score,
  $h_0$ is the function satisfying \cref{eq:outcome-bridge-function},
  $\bar{h}_0(A, X)$ is defined as $\E\bs{h_0(W, S, X) \mid A, X, G = \text{E}}$,
  and $q_{a, 0}$'s are the functions satisfying \cref{eq:surrogate-bridge-function}.
  If any of the following conditions on nuisance functions holds:
  \begin{enumerate}
    \item $h = h_0$, $\bar{h} = \bar{h}_0$;
    \item $h = h_0$, $e = e_0$;
    \item $q_a = q_{a, 0}$ for each $a \in \{0, 1\}$, $e = e_0$;
    \item $q_a = q_{a, 0}$ for each $a \in \{0, 1\}$, $\bar{h} = \bar{h}^{(h)}$;
  \end{enumerate}
  then
  \begin{align}
    \label{eq:mutiply-robust-identification}
    \tau_0 = \E\bs{\varphi(D; \eta)},
  \end{align}
  where $\varphi(D; \eta) \equiv \varphi_{\text{E}}(D; \eta) + \varphi_{\text{O}}(D; \eta)$
  is defined as
  \begin{align*}
    \varphi_{\text{E}}(D; \eta)
    &\equiv \frac{\mathbf{1}_{\bc{G = \text{E}}}}{\pi_0} \bp{\frac{(A - e(X))\bp{h(W, S, X) - \bar{h}(A, X)}}{e(X)(1 - e(X))} + \bar{h}(1, X) - \bar{h}(0, X)} \\
    \varphi_{\text{O}}(D; \eta)
    &\equiv \frac{\mathbf{1}_{\bc{G = \text{O}}}}{1 - \pi_0} \bp{\bp{q_1(Z, S, X) - q_0(Z, S, X)} (Y - h(W, S, X))}.
  \end{align*}
\end{theorem}

\begin{proof}
  See \cref{sec:proof-of-thm-mutiply-robust-identification}.
\end{proof}

\begin{remark}
  \cref{thm:mutiply-robust-identification}
  is multiply robust
  in the sense that
  it identifies the ATE $\tau_0$
  if any one of the four sets of nuisance functions
  is correctly specified.
  This identification result
  is motivated by
  the EIF of $\tau_0$
  under a certain nonparametric model (see \cref{sec:estimation-and-inference}),
  extending the influence function-based identification
  in \cite{athey2025surrogate}
  to accommodate unobserved confounding.
  See \cref{sec:connection-to-influence-function-identification}
  for more details.
\end{remark}

\section{Estimation and Inference} \label{sec:estimation-and-inference}

\subsection{Cross-Fitted Estimators} \label{sec:cross-fitted-estimators}

In this section, 
we propose a generic estimation procedure
for the ATE $\tau_0$.
This approach
builds upon the identification results in \cref{sec:identification}
and employs
the cross-fitting technique from
\cite{chernozhukov2018double}
to 
reduce overfitting bias
when estimating nuisance parameters.

\begin{definition}[Cross-Fitted Estimators]
  \label{def:cross-fitted-estimators}
  Fix an integer $K \geq 2$.
  Let $\mathcal{I} \equiv \{1, \dots, N\}$ be the index set,
  where $N$ is the total sample size.
  Define $\mathcal{I}_{\text{E}} \equiv \{i \in \mathcal{I}: G_i = \text{E}\}$
  and $\mathcal{I}_{\text{O}} \equiv \{i \in \mathcal{I}: G_i = \text{O}\}$
  as the index sets for the experimental and observational samples, respectively.
  Randomly partition $\mathcal{I}_{\text{E}}$ into $K$ disjoint subsets
  $\{\mathcal{I}_{\text{E}, k}\}_{k = 1}^K$.
  Similarly, randomly partition $\mathcal{I}_{\text{O}}$
  into $K$ disjoint subsets $\{\mathcal{I}_{\text{O}, k}\}_{k = 1}^K$.
  Define $\mathcal{D}_{k} \equiv \{D_i: i \in \mathcal{I}_{\text{E}, k} \cup \mathcal{I}_{\text{O}, k}\}$
  as the evaluation set for fold $k$,
  and $\mathcal{D}_{-k} \equiv \mathcal{D} \setminus \mathcal{D}_{k}$
  as the training set for fold $k$.
  For each fold $k$,
  estimate the nuisance functions
  $\hat{\eta}_{k} \equiv (\hat{e}_{k}, \hat{h}_{k}, \hat{\bar{h}}_{k}, \hat{q}_{0, k}, \hat{q}_{1, k})$
  using the data excluding fold $k$, i.e., $\mathcal{D}_{-k}$,
  where $\hat{\bar h}_k$ is fitted by regressing $\hat h_k(W,S,X)$ on $(A,X)$ using $\mathcal D_{\mathrm{E},-k}$.
  The cross-fitted estimators are defined as
  \begin{align*}
    \hat{\tau}_{\text{OB-OR}} 
      &\equiv \frac{1}{N_{\text{E}}} \sum_{k = 1}^K \sum_{i \in \mathcal{I}_{\text{E}, k}} 
      \bp{\hat{\bar{h}}_{k}(1, X_i) - \hat{\bar{h}}_{k}(0, X_i)}, \\
    \hat{\tau}_{\text{OB-IPW}}
      &\equiv \frac{1}{N_{\text{E}}} \sum_{k = 1}^K \sum_{i \in \mathcal{I}_{\text{E}, k}} 
       \bp{\frac{A_i \hat{h}_{k}(W_i, S_i, X_i)}{\hat{e}_{k}(X_i)} - \frac{(1 - A_i) \hat{h}_{k}(W_i, S_i, X_i)}{1 - \hat{e}_{k}(X_i)}}, \\
    \hat{\tau}_{\text{SB}}
      &\equiv \frac{1}{N_{\text{O}}} \sum_{k = 1}^K \sum_{i \in \mathcal{I}_{\text{O}, k}} 
      \bp{\hat{q}_{1, k}(Z_i, S_i, X_i) - \hat{q}_{0, k}(Z_i, S_i, X_i)} Y_i, \\
    \hat{\tau}_{\text{MR}}
      &\equiv \frac{1}{N_{\text{E}}} \sum_{k = 1}^K \sum_{i \in \mathcal{I}_{\text{E}, k}} 
      \bp{\frac{(A_i - \hat{e}_{k}(X_i))(\hat{h}_{k}(W_i, S_i, X_i) - \hat{\bar{h}}_{k}(A_i, X_i))}{\hat{e}_{k}(X_i)(1 - \hat{e}_{k}(X_i))}} \\
      &\quad + \frac{1}{N_{\text{E}}} \sum_{k = 1}^K \sum_{i \in \mathcal{I}_{\text{E}, k}} \bp{\hat{\bar{h}}_{k}(1, X_i) - \hat{\bar{h}}_{k}(0, X_i)} \\
      &\quad + \frac{1}{N_{\text{O}}} \sum_{k = 1}^K \sum_{i \in \mathcal{I}_{\text{O}, k}} 
      \bp{(\hat{q}_{1, k}(Z_i, S_i, X_i) - \hat{q}_{0, k}(Z_i, S_i, X_i)) (Y_i - \hat{h}_{k}(W_i, S_i, X_i))}.
  \end{align*}
\end{definition}

\subsection{Asymptotic Properties} \label{sec:asymptotic-properties}

Before stating the main theorems,
we introduce some additional definitions and assumptions.
First,
we define the linear operator $T: L_2(W, S, X) \to L_2(Z, S, X)$
as 
\[
  (T h)(Z, S, X) = \E\bs{h(W, S, X) \mid Z, S, X, G = \text{O}},
\]
and its adjoint operator
$T^*: L_2(Z, S, X) \to L_2(W, S, X)$
as
\[
  (T^* q_a)(W, S, X) = \E\bs{q_a(Z, S, X) \mid W, S, X, G = \text{O}}.
\]

\subsubsection{Consistency}

Following the standard approach in \cite{imbensLongtermCausalInference2025a},
we impose the following convergence rate conditions
on the nuisance estimators to establish consistency.

\begin{assumption}[Convergence Rates]
  \label{asn:convergence-rates}
  Let $\| \cdot \|_{2}$ denote the $L_2$ norm.
  For any square-integrable function $g(W, S, X)$, 
  let $\tilde{\bar{h}}^{(g)}(A, X)$ denote the population limit of the
  estimator $\hat{\bar{h}}$ trained with pseudo-outcome $g(W, S, X)$ and predictors $(A, X)$ 
  in the experimental sample.
  For each fold $k$,
  the nuisance estimators
  satisfy the following convergence rates:
  \begin{align*}
    \|\hat{h}_{k} - \tilde{h}\|_{2} &= O_P(\delta_{h, N}), \quad
    \|T(\hat{h}_{k} - \tilde{h})\|_{2} = O_P(\rho_{h, N}), \quad
    \|\hat{\bar h}_{k} - \tilde{\bar{h}}^{(\hat{h}_k)}\|_2 = O_P(\delta_{\bar h,N}), \\
    \|\hat{q}_{a, k} - \tilde{q}_{a}\|_{2} &= O_P(\delta_{q, N}), \quad
    \|T^*(\hat{q}_{a, k} - \tilde{q}_{a})\|_{2} = O_P(\rho_{q, N}), \quad
    \|\hat e_{k} - \tilde{e}\|_{2} = O_P(\delta_{e, N}),
  \end{align*}
  where $\delta_{h, N}$, $\rho_{h, N}$, $\delta_{\bar{h}, N}$, $\delta_{q, N}$, $\rho_{q, N}$, $\delta_{e, N}$
  are sequences converging to zero as $N \to \infty$,
  and the mapping $g \mapsto \tilde{\bar{h}}^{(g)}$ is Lipschitz continuous with respect to the $L_2$ norm.
\end{assumption}

\begin{theorem}[Consistency]
\label{thm:consistency}
  \begin{enumerate}
    \item[]
    \item If the conditions in \cref{thm:outcome-bridge-identification},
      \cref{asn:convergence-rates},
      $\tilde{h} = h_0$ and $\tilde{\bar{h}}^{(h_0)} = \bar{h}_0$,
      then $\hat{\tau}_{\text{OB-OR}}$ is consistent for $\tau_0$.

    \item If the conditions in \cref{thm:outcome-bridge-identification},
      \cref{asn:convergence-rates},
      $\tilde{h} = h_0$ and $\tilde{e} = e_0$ hold,
      then $\hat{\tau}_{\text{OB-IPW}}$ is consistent for $\tau_0$.

    \item If the conditions in \cref{thm:surrogate-bridge-identification},
      \cref{asn:convergence-rates},
      $\tilde{q}_{a} = q_{a, 0}$ for each $a \in \{0, 1\}$
      and $\tilde{e} = e_0$ hold,
      then $\hat{\tau}_{\text{SB}}$ is consistent for $\tau_0$.
      
    \item If \cref{asn:convergence-rates} holds, and any of the following sets of conditions is satisfied:
      \begin{enumerate}
        \item $\tilde{h} = h_0$ and $\tilde{\bar{h}}^{(h_0)} = \bar{h}_0$;
        \item $\tilde{h} = h_0$ and $\tilde{e} = e_0$;
        \item $\tilde{q}_{a} = q_{a, 0}$ for each $a \in \{0, 1\}$ and $\tilde{e} = e_0$;
        \item $\tilde{q}_{a} = q_{a, 0}$ for each $a \in \{0, 1\}$ and $\tilde{\bar{h}}^{(\tilde{h})}(A, X) = \E[\tilde{h} \mid A, X, G=\text{E}]$,
      \end{enumerate}
      then $\hat{\tau}_{\text{MR}}$ is consistent for $\tau_0$.
  \end{enumerate}
\end{theorem}

\begin{proof}
  See \cref{sec:proof-of-thm-consistency}.
\end{proof}

\subsubsection{Asymptotic Normality}

While consistency requires only that the nuisance estimators converge to their true values, 
establishing asymptotic normality requires stricter conditions on their convergence rates.
Specifically, we need to ensure that the second-order bias terms vanish faster than $N^{-1/2}$,
ensuring that the asymptotic distribution of $\hat{\tau}_{\text{MR}}$
is dominated by the first-order term.
We impose the following assumption on the product rates of the nuisance estimators.

\begin{assumption}[Product Rates of Convergence]
  \label{asn:product-rates}
  Suppose that \cref{asn:convergence-rates} holds.
  The convergence rates satisfy the following conditions:
  \begin{align}
    \delta_{e, N} \delta_{\bar{h}, N} = o_P(N^{-1/2}),
    \label{eq:rate-cond-experimental} \\
    \min \left\{ \delta_{q, N} \rho_{h, N}, \; \rho_{q, N} \delta_{h, N} \right\} = o_P(N^{-1/2}).
    \label{eq:rate-cond-observational}
  \end{align}
\end{assumption}

\begin{remark}
  \cref{eq:rate-cond-experimental}
  is standard in the DML literature
  \citep{chernozhukov2018double},
  requiring the product of the convergence rates of the propensity score
  and the pseudo-outcome regression
  that uses the estimated outcome bridge function as the outcome
  to be $o_P(N^{-1/2})$.
  \cref{eq:rate-cond-observational} is specific to the proximal inference setting.
  Following \cite{imbensLongtermCausalInference2025a},
  we distinguish between two types of convergence rates:
  \emph{strong-metric errors} ($\delta_{h,N}$ and $\delta_{q,N}$)
  measure $L_2$ deviation from the true bridge functions,
  and \emph{weak-metric errors} ($\rho_{h,N}$ and $\rho_{q,N}$)
  measure violations of the bridge equations via operators $T$ and $T^*$.
  Because bridge-function estimation amounts to solving potentially ill-posed conditional moment equations,
  an estimator can exhibit a slow strong-metric rate while still achieving a fast weak-metric rate
  (e.g., via Tikhonov regularization or sieve methods).
  \cref{eq:rate-cond-observational} exploits this asymmetry,
  requiring only that the product of \emph{one} strong-metric error and \emph{one} weak-metric error
  be $o_P(N^{-1/2})$,
  a strictly weaker condition than requiring both strong rates 
  to satisfy $\delta_{h,N} \delta_{q,N} = o_P(N^{-1/2})$.
\end{remark}

The following theorem establishes the asymptotic normality of the 
cross-fitted estimator $\hat{\tau}_{\text{MR}}$
when the bridge functions and propensity score 
converge to their true values and the pseudo-outcome regression for $\hat{\bar h}_k$ 
is correctly specified.

\begin{theorem}[Asymptotic Normality]
  \label{thm:asymptotic-normality}
  Suppose \cref{asn:convergence-rates} holds with the correct limits
  \[
    \tilde{h} = h_0, \quad \tilde{e} = e_0, \quad \tilde{q}_a = q_{a,0} \text{ for } a \in \{0, 1\}, \quad \text{and} \quad
  \tilde{\bar{h}}^{(h_0)} = \bar{h}_0.
  \]
  If \cref{asn:product-rates} also holds,
  then the cross-fitted estimator $\hat{\tau}_{\text{MR}}$ is asymptotically normal:
  \begin{align*}
    \sqrt{N}(\hat{\tau}_{\text{MR}} - \tau_0) \xrightarrow{d} \mathcal{N}(0, V),
  \end{align*}
\end{theorem}

\begin{proof}
  See \cref{sec:proof-of-thm-asymptotic-normality}.
\end{proof}

An immediate consequence of \cref{thm:asymptotic-normality}
is that
we can construct
a consistent estimator for the asymptotic variance $V$
using the plug-in approach:
\begin{align*}
  \hat{V} 
  &= \frac{N}{N_{\text{E}}^2} \sum_{k = 1}^K \sum_{i \in \mathcal{I}_{\text{E}, k}}
  \bp{\frac{(A_i - \hat{e}_{k}(X_i))(\hat{h}_{k}(W_i, S_i, X_i) - \hat{\bar{h}}_{k}(A_i, X_i))}{\hat{e}_{k}(X_i)(1 - \hat{e}_{k}(X_i))} + \hat{\bar{h}}_{k}(1, X_i) - \hat{\bar{h}}_{k}(0, X_i) - \hat{\tau}_{\text{MR}}}^2 \\
  &\quad + \frac{N}{N_{\text{O}}^2} \sum_{k = 1}^K \sum_{i \in \mathcal{I}_{\text{O}, k}}
  \bp{(\hat{q}_{1, k}(Z_i, S_i, X_i) - \hat{q}_{0, k}(Z_i, S_i, X_i)) (Y_i - \hat{h}_{k}(W_i, S_i, X_i))}^2.
\end{align*}

The following theorem establishes the consistency of this variance estimator
and the validity of the corresponding confidence interval (CI).

\begin{theorem}[Consistency of Variance Estimator]
  \label{thm:consistency-of-variance-estimator}
  Under the conditions of \cref{thm:asymptotic-normality},
  the variance estimator $\hat{V}$ is consistent for $V$:
  \begin{align*}
    \hat{V} \xrightarrow{p} V.
  \end{align*}
  Consequently, the asymptotic $(1-\alpha)$ CI for $\tau_0$ is given by
  \begin{align*}
    \text{CI}_{1 - \alpha} 
    = \left( \hat{\tau}_{\text{MR}} - \Phi^{-1}(1 - \alpha/2) \sqrt{\hat{V}/N}, \quad
    \hat{\tau}_{\text{MR}} + \Phi^{-1}(1 - \alpha/2) \sqrt{\hat{V}/N} \right),
  \end{align*}
  satisfying $\lim_{N \to \infty} P(\tau_0 \in \text{CI}_{1 - \alpha}) = 1 - \alpha$,
  where $\Phi^{-1}(\cdot)$ is the quantile function of the standard normal distribution.
\end{theorem}

\begin{proof}
  See \cref{sec:proof-of-thm-consistency-of-variance-estimator}.
\end{proof}

\subsubsection{Semiparametric Efficiency}

We now establish the local semiparametric efficiency of 
our cross-fitted estimator $\hat{\tau}_{\text{MR}}$
under a statistical model $\mathcal{P}$
that imposes no restrictions
on the observed data law
other than the existence of an outcome bridge function $h_0$
(\cref{asn:outcome-bridge-function})

Following \cite{cuiSemiparametricProximalCausal2024} and \cite{imbensLongtermCausalInference2025a},
to ensure the EIF is well-defined,
we impose the following surjectivity condition,
which guarantees
the uniqueness of the bridge functions.

\begin{assumption}[Surjectivity]
  \label{asn:surjectivity}
  The linear operators $T$ and $T^*$
  are surjective.
\end{assumption}

\begin{proposition}[Uniqueness of Bridge Functions]
  \label{prp:surjectivity-uniqueness}
  Under \cref{asn:surjectivity},
  the outcome bridge function $h_0$
  in \cref{eq:outcome-bridge-function}
  and the surrogate bridge functions $q_{a, 0}$'s
  in \cref{eq:surrogate-bridge-function}
  are unique almost surely.
\end{proposition}

\begin{proof}
  See \cref{sec:proof-of-surjectivity-uniqueness}.
\end{proof}

The following theorem characterizes the
EIF for the ATE $\tau_0$
and the corresponding semiparametric efficiency bound.

\begin{theorem}[EIF]
  \label{thm:eif}
  Consider the semiparametric model $\mathcal{P}$
  evaluated at the submodel
  where the \cref{asn:surrogate-bridge-function,asn:surjectivity} hold.
  The EIF for the ATE $\tau_0$
  is given by
  \begin{align*}
    &\mathrel{\phantom{=}} \IF(D) \\
    &\equiv \frac{\mathbf{1}_{\bc{G = \text{E}}}}{\pi_0} \bp{\frac{(A - e_0(X))\bp{h_0(W, S, X) - \bar{h}_0(A, X)}}{e_0(X)(1 - e_0(X))} + \bar{h}_0(1, X) - \bar{h}_0(0, X) - \tau_0} \\
    &\quad + \frac{\mathbf{1}_{\bc{G = \text{O}}}}{1 - \pi_0} \bp{\bp{q_{1, 0}(Z, S, X) - q_{0, 0}(Z, S, X)} (Y - h_0(W, S, X))}.
  \end{align*}
  Therefore,
  the corresponding semiparametric local efficiency bound 
  for estimating $\tau_0$ is $V_\text{eff} = \E\bs{\IF(D)^2}$.
\end{theorem}

\begin{proof}
  See \cref{sec:proof-of-thm-eif}.
\end{proof}

A direct consequence of \cref{thm:asymptotic-normality} and \cref{thm:eif}
is that
our cross-fitted estimator $\hat{\tau}_{\text{MR}}$
achieves the semiparametric local efficiency bound
as stated in the following theorem.

\begin{theorem}[Semiparametric Efficiency]
  \label{thm:efficiency}
  Suppose the conditions of \cref{thm:asymptotic-normality}
  and \cref{asn:surjectivity} hold.
  Then,
  \begin{align*}
    \sqrt{N} (\hat{\tau}_{\text{MR}} - \tau_0)
    \xrightarrow{d} \mathcal{N}(0, V_\text{eff}),
  \end{align*}
  where $V_\text{eff}$ is the semiparametric efficiency bound
  derived in \cref{thm:eif}.
\end{theorem}

\begin{proof}
  See \cref{sec:proof-of-thm-efficiency}.
\end{proof}

\section{Real Data Application} \label{sec:real-data-application}

Following the seminal approach of \cite{lalonde1986evaluating},
we exploit the experimental data to
construct a ``ground truth'' benchmark 
against which we evaluate 
the performance of different methods.
We analyze data from the JC program,
a large-scale randomized evaluation of job training programs
for disadvantaged youth in the United States
\citep{schochet2001national,schochet2008does}.

\paragraph{Design}

Our evaluation design proceeds in two steps. 
First, we randomly split the experimental dataset 
into two disjoint subsets: an ``Experimental Sample'' ($\text{E}$) 
where the long-term outcome ($Y$) is masked, 
and an ``Observational Sample'' ($\text{O}$) 
where the treatment assignment ($A$) is masked.\footnote{
  The dataset for this application
  is publicly available at the \texttt{causalweight} R package
  \citep{causalweight}.
}
This setup mimics the data combination problem 
while ensuring, by construction, 
that there is no selection bias between the two samples 
(i.e., comparability or transportability assumptions hold).
This allows us to isolate 
the performance of the estimators 
in handling the failure of the surrogacy assumption
or the presence of unobserved confounding.

\paragraph{Variables}

We define the treatment ($A$) 
as random assignment and 
the long-term outcomes ($Y$) 
as weekly earnings and proportion of weeks employed
in the fourth year. 
The surrogate vector ($S$)
includes earnings and proportion of weeks employed
in years 2 and 3. 
To address potential confounding, 
we select possession of a general education development (GED) degree ($W$) 
and English mother tongue ($Z$) as proxies.
These variables are determined pre-treatment 
and likely provide information about unobserved confounders,
such as cognitive skills or long-term career orientation.
We also include baseline covariates ($X$),
including age, gender, race, education level, and pre-program earnings.

\paragraph{Methods Compared}

We compare four methods for estimating the ATE.
The first method
is the RCT benchmark,
which estimates the ATE
directly from the experimental sample
using a linear regression of $Y$ on $(1, A, X)$.
The second method
is the standard surrogate index estimator
\citep{athey2025surrogate},
which first fits a regression of $Y$ on $(1, S, X)$
in the observational sample,
then predicts $Y$ in the experimental sample,
and finally estimates the ATE
using a linear regression of the predicted $Y$ on $(1, A, X)$.
The third method
is similar to the second method
but includes the proxy variables ($W$ and $Z$)
as additional covariates
in both regressions.
The fourth method
is our proposed method,
which first estimates an linear outcome bridge function
$h(W, S, X; \beta)$
in the observational sample,
then predicts $Y$ in the experimental sample
using the fitted $h$,
and finally estimates the ATE
using a linear regression of the predicted $Y$ on $(1, A, X)$.

\paragraph{Results}

\begin{table}[ht]
  \footnotesize
  \centering
  \caption{Estimates of the 4-year effects of the JC program on earnings and proportion of weeks employed.}
  \label{tab:jc-ate}
  
\begin{tabular}{lccc}
\toprule
 & Estimate & Standard Error & $p$-Value\\
\midrule
\addlinespace[0.3em]
\multicolumn{4}{l}{\textit{Panel A: Weekly Earnings}}\\
\hspace{1em}RCT & 15.30 & 5.49 & 0.005\\
\hspace{1em}Naive (1) & 7.05 & 3.70 & 0.057\\
\hspace{1em}Naive (2) & 7.30 & 3.69 & 0.048\\
\hspace{1em}Ours & 16.43 & 7.92 & 0.038\\
\addlinespace[0.3em]
\multicolumn{4}{l}{\textit{Panel B: Proportion of Weeks Employed}}\\
\hspace{1em}RCT & 2.87 & 1.14 & 0.012\\
\hspace{1em}Naive (1) & 0.80 & 0.67 & 0.233\\
\hspace{1em}Naive (2) & 0.84 & 0.67 & 0.210\\
\hspace{1em}Ours & 3.50 & 1.96 & 0.074\\
\bottomrule
\end{tabular}

  \caption*{
    \textit{Note:}
    The standard errors
    are computed using the plug-in sandwich estimator
    based on the joint GMM framework.
    The $p$-values
    are based on two-sided Wald tests
    using normal approximations.
  }
\end{table}

\cref{tab:jc-ate}
presents
the estimates of the treatment effects 
in the experimental sample.
While the true benchmark 
shows a significant positive effect,
the standard surrogate index method
substantially underestimates
the effect for both outcomes,
and including the proxies
$W$ and $Z$
as additional covariates
does not improve the estimates.
In contrast,
our proposed method
that accounts for unobserved confounding
via the proxy variables
yields estimates
much closer to the benchmark results,
albeit with larger standard errors.

\paragraph{Discussion}

As found in \cite{schochet2008does},
program participation
temporarily suppresses short-term earnings and employment
because participants
spend time in training instead of working.
Theerefore, the substantial underestimation by the standard surrogate index 
likely arises from unobserved confounding, 
where latent ability $U$ positively affects both short-term ($S$) and long-term ($Y$) outcomes.
This confounding creates a positive bias that attenuates the true negative structural relationship between 
$S$ and $Y$, where sacrificing short-term earnings represents an investment in future growth.
Consequently,
the standard method underestimates the true long-term effect of the program.
By leveraging proxies to adjust for $U$, 
our method distinguishes between low earnings due to investment versus low ability, 
thereby recovering the experimental benchmark.

\paragraph{Diagnostic Analysis}

We perform a diagnostic analysis 
on the experimental sample 
to investigate the validity 
of the standard surrogacy assumption 
and the role of unobserved confounding. 
The results are summarized in \cref{tab:jc-diagnostic}.

First,
since we have access to the long-term outcome $Y$
in the experimental sample,
we can directly test the surrogacy assumption
by regressing $Y$ on $(1, A, S, X)$
and examining the coefficient of $A$.
The first rows from Panel A
and Panel B
of \cref{tab:jc-diagnostic}
show that
the coefficient of $A$
is statistically significant
for both outcomes,
indicating
a potential violation
of the surrogacy assumption.
Second,
to assess whether the proxies
$W$ and $Z$
help mitigate unobserved confounding,
we further fit an IV regression,
where the first stage regresses $W$ on $(1, Z, A, S, X)$
and the second stage regresses $Y$
on $(1, A, S, X, \hat W)$,
with $Z$ acting as the instrument for $W$.
In this setup,
$\hat W$
captures the variation coming from the unobserved confounders $U$,
so including it in the regression
helps adjust for unobserved confounding.
The second rows from Panel A
and Panel B
of \cref{tab:jc-diagnostic}
show that
the coefficient of $A$
becomes statistically insignificant
for both outcomes
after adjusting for $\hat W$,
suggesting
that
the proxies
$W$ and $Z$
effectively mitigate
unobserved confounding,
which results in
the improved performance
in \cref{tab:jc-ate}.

\begin{table}[ht]
  \footnotesize
  \centering
  \caption{Diagnostic analysis on the experimental sample of the JC program.}
  \label{tab:jc-diagnostic}
  
\begin{tabular}{lccc}
\toprule
 & Coefficient on $A$ & Standard Error & $p$-Value\\
\midrule
\addlinespace[0.3em]
\multicolumn{4}{l}{\textit{Panel A: Weekly Earnings}}\\
\hspace{1em}OLS Regression & 8.603 & 4.286 & 0.045\\
\hspace{1em}IV Regression & -0.125 & 9.465 & 0.989\\
\addlinespace[0.3em]
\multicolumn{4}{l}{\textit{Panel B: Proportion of Weeks Employed}}\\
\hspace{1em}OLS Regression & 2.117 & 0.933 & 0.023\\
\hspace{1em}IV Regression & 0.546 & 1.848 & 0.768\\
\bottomrule
\end{tabular}

  \caption*{
    \textit{Note:}
    The standard errors
    are computed using heteroskedasticity-robust standard errors.
    The $p$-values
    are based on two-sided $t$-tests.
  }
\end{table}

\section{Conclusion} \label{sec:conclusion}

We study the challenging problem 
of identifying and estimating long-term treatment effects 
by combining experimental and observational data 
in the presence of unobserved confounding.
The standard surrogate index methods rely on
strong assumptions that rule out such confounding,
leading to potentially biased estimates in many empirical settings.

We develop new methods
for relaxing these stringent assumptions.
Our methods rely on proxy variables
that are informative about the unobserved confounders
and satisfy certain conditional independence properties.
We establish nonparametric identification results
based on bridge functions,
which generalize
existing identification strategies that assume no unobserved confounding.
We further propose
estimation and inference procedures
and establish their asymptotic properties,
including consistency, asymptotic normality, and semiparametric efficiency.

Applying our methods to
a real-world job training program evaluation,
we find that
the standard surrogate index approach
substantially underestimates
the true long-term effects,
while our proposed methods
that account for unobserved confounding
via proxy variables
yield estimates
much closer to the benchmark results
from the experimental data.

\phantomsection
\addcontentsline{toc}{section}{References}
\bibliographystyle{apalike}
\bibliography{references}

\clearpage
\appendix

\section{Proofs}

\subsection{Proof of \cref{thm:outcome-bridge-identification}} \label{sec:proof-of-thm-outcome-bridge-identification}

We prove each part of \cref{thm:outcome-bridge-identification} in order.

\paragraph{Part 1}

Suppose that the function $h_0$ satisfies \cref{eq:outcome-bridge-function}:
\begin{align*}
  \E\bs{Y \mid Z, S, X, G = \text{O}}
  = \E\bs{h_0(w, S, X) \mid Z, S, X, G = \text{O}}.
\end{align*}
Note that
\begin{align*}
  &\mathrel{\phantom{=}} \E\bs{Y \mid Z, S, X, G = \text{O}} \\
  &= \E\bs{\E\bs{Y \mid U, Z, S, X, G = \text{O}} \mid Z, S, X, G = \text{O}} \tag{Iterated expectation} \\
  &= \E\bs{\E\bs{Y \mid U, S, X, G = \text{O}} \mid Z, S, X, G = \text{O}} \tag{$Y \ind Z \mid (X, U, S, G)$} \\
  &= \E\bs{h_0(W, S, X) \mid Z, S, X, G = \text{O}} \tag{\cref{eq:outcome-bridge-function}} \\
  &= \E\bs{\E\bs{h_0(W, S, X) \mid U, Z, S, X, G = \text{O}} \mid Z, S, X, G = \text{O}} \tag{Iterated expectation} \\
  &= \E\bs{\E\bs{h_0(W, S, X) \mid U, X, G = \text{O}} \mid Z, S, X, G = \text{O}}. \tag{$W \ind (S, Z) \mid (X, U, G)$}
\end{align*}
By the completeness condition in \cref{asn:completeness-outcome},
comparing the third and last lines above
yields
\begin{align*}
  \E\bs{Y \mid U, S, X, G = \text{O}}
  = \E\bs{h_0(W, S, X) \mid U, X, G = \text{O}}.
\end{align*}
This completes the proof of Part 1.

\paragraph{Part 2}

Under \cref{asn:transportability},
we have
\begin{align*}
  \E\bs{Y \mid U, S, X, G = \text{E}}
  &= \E\bs{Y \mid U, S, X, G = \text{O}} \tag{$Y \ind G \mid (S, X, U)$} \\
  &= \int_{\mathcal{W}} h_0(w, S, X) \, d F_{W \mid U, X, G = \text{O}}(w) \tag{Part 1} \\
  &= \int_{\mathcal{W}} h_0(w, S, X) \, d F_{W \mid U, X, G = \text{E}}(w) \tag{$W \ind G \mid (X, U)$} \\
  &= \E\bs{h_0(W, S, X) \mid U, X, G = \text{E}}.
\end{align*}
This completes the proof of Part 2.

\paragraph{Part 3}

Note that \cref{asn:unconfoundedness-experimental,asn:no-direct-effect}
imply that
$A \ind Y(s) \mid (X, U, G = \text{E})$
for all $s \in \mathcal{S}$.
\cref{asn:unconfoundedness-experimental}
also implies that
$S \ind Y(s) \mid (A, X, U, G = \text{E})$.
Applying the contraction property of conditional independence
\citep[e.g., Section 1.1.5 of][]{pearl2009causality},
we have $(A, S) \ind Y(s) \mid (X, U, G = \text{E})$,
which further implies that $A \ind Y(s) \mid (S, X, U, G = \text{E})$
by the weak union property of conditional independence.
It follows that $A \ind Y(S) \mid (S, X, U, G = \text{E})$.
By the consistency assumption of potential outcomes,
we have 
\begin{align}
  A \ind Y \mid (S, X, U, G = \text{E}). \label{eq:independence-for-ob-proof}
\end{align}

Now we can identify the mean potential outcome as follows:
\begin{align}
  &\mathrel{\phantom{=}} \E\bs{Y(a) \mid X, G = \text{E}} \notag \\
  &= \E\bs{Y(a) \mid A = a, X, G = \text{E}} \label{eq:ob-proof-1} \\
  &= \E\bs{Y(S(a)) \mid A = a, X, G = \text{E}} \label{eq:ob-proof-2} \\
  &= \E\bs{Y(S) \mid A = a, X, G = \text{E}} \label{eq:ob-proof-3} \\
  &= \int_{\mathcal{S}} \E\bs{Y(s) \mid S = s, A = a, X, G = \text{E}}
      \, d F_{S \mid a, X, \text{E}}(s) \label{eq:ob-proof-4} \\
  &= \int_{\mathcal{S}} \int_{\mathcal{U}}
      \E\bs{Y(s) \mid U = u, S = s, A = a, X, G = \text{E}}
      \, d F_{U \mid s, a, X, G = \text{E}}(u)
      \, d F_{S \mid a, X, \text{E}}(s) \label{eq:ob-proof-5} \\
  &= \int_{\mathcal{S}} \int_{\mathcal{U}}
      \E\bs{Y \mid U = u, S = s, A = a, X, G = \text{E}}
      \, d F_{U \mid s, a, X, G = \text{E}}(u)
      \, d F_{S \mid a, X, \text{E}}(s) \label{eq:ob-proof-6} \\
  &= \int_{\mathcal{S}} \int_{\mathcal{U}}
      \E\bs{Y \mid U = u, S = s, X, G = \text{E}}
      \, d F_{U \mid s, a, X, G = \text{E}}(u)
      \, d F_{S \mid a, X, \text{E}}(s) \label{eq:ob-proof-7} \\
  &= \int_{\mathcal{S}} \int_{\mathcal{U}} \int_{\mathcal{W}}
      h_0(w, s, X)
      \, d F_{W \mid u, X, G = \text{E}}(w)
      \, d F_{U \mid s, a, X, G = \text{E}}(u)
      \, d F_{S \mid a, X, \text{E}}(s) \label{eq:ob-proof-8} \\
  &= \int_{\mathcal{S}} \int_{\mathcal{W}}
      h_0(w, s, X)
      \int_{\mathcal{U}}
      d F_{W \mid u, X, G = \text{E}}(w)
      \, d F_{U \mid s, a, X, G = \text{E}}(u)
      \, d F_{S \mid a, X, \text{E}}(s) \label{eq:ob-proof-9} \\
  &= \int_{\mathcal{S}} \int_{\mathcal{W}}
      h_0(w, s, X)
      \, d F_{W \mid u, s, a, X, G = \text{E}}(w)
      \, d F_{U \mid s, a, X, G = \text{E}}(u)
      \, d F_{S \mid a, X, \text{E}}(s) \label{eq:ob-proof-10} \\
  &= \int_{\mathcal{S}} \int_{\mathcal{W}}
      h_0(w, s, X)
      \, d F_{W \mid s, a, X, G = \text{E}}(w)
      \, d F_{S \mid a, X, \text{E}}(s) \label{eq:ob-proof-11} \\
  &= \E_{S \mid A = a, X, G = \text{E}} \E\bs{h_0(W, S, X) \mid S, A = a, X, G = \text{E}} \label{eq:ob-proof-12} \\
  &= \E\bs{h_0(W, S, X) \mid A = a, X, G = \text{E}}. \label{eq:ob-proof-13}
\end{align}
where
\cref{eq:ob-proof-1} follows from 
  $A \ind Y(a) \mid X, G = \text{E}$ (\cref{asn:unconfoundedness-experimental}),
\cref{eq:ob-proof-2} follows from the composition assumption of potential outcomes,
\cref{eq:ob-proof-3} follows from $S(a) = S$ when $A = a$,
\cref{eq:ob-proof-4,eq:ob-proof-5} follow from the law of total probability,
\cref{eq:ob-proof-6} uses the consistency assumption of potential outcomes,
\cref{eq:ob-proof-7} applies \cref{eq:independence-for-ob-proof},
\cref{eq:ob-proof-8} applies Part 2,
\cref{eq:ob-proof-9} exchanges the order of integration,
\cref{eq:ob-proof-10} follows from $W \ind (A, S) \mid (X, U, G = \text{E})$ (\cref{asn:availability-of-proxies}),
\cref{eq:ob-proof-11} follows from the law of total probability,
\cref{eq:ob-proof-12} rewrites the expression in expectation form,
and \cref{eq:ob-proof-13} simplifies the expression using the law of iterated expectation.
Finally, we have
\begin{align*}
  &\mathrel{\phantom{=}}
  \E\bs{\frac{A h_0(W, S, X)}{e_0(X)} - \frac{(1 - A) h_0(W, S, X)}{1 - e_0(X)} \;\middle|\; G = \text{E}} \\
  &= \E\bs{\E\bs{\frac{A h_0(W, S, X)}{e_0(X)} - \frac{(1 - A) h_0(W, S, X)}{1 - e_0(X)} \;\middle|\; X, G = \text{E}} \;\middle|\; G = \text{E}} \\
  &= \E\bs{\E\bs{h_0(W, S, X) \mid A = 1, X, G = \text{E}} - \E\bs{h_0(W, S, X) \mid A = 0, X, G = \text{E}} \;\middle|\; G = \text{E}} \\
  &= \tau_0.
\end{align*}
This completes the proof of Part 3.

\subsection{Proof of \cref{thm:surrogate-bridge-identification}} \label{sec:proof-of-thm-surrogate-bridge-identification}

We prove each part of \cref{thm:surrogate-bridge-identification} in order.

\paragraph{Part 1}

Suppose that the function $q_{a, 0}$ satisfies \cref{eq:surrogate-bridge-function}:
\begin{align*}
  \E\bs{q_{a, 0}(Z, S, X) \mid W, S, X, G = \text{O}}
  = \frac{f(W, S \mid A = a, X, G = \text{E}) f(X \mid G = \text{E})}{f(W, S, X \mid G = \text{O})}.
\end{align*}
The left-hand side of \cref{eq:surrogate-bridge-function} can be expanded as follows:
\begin{align*}
  &\mathrel{\phantom{=}} 
  \E\bs{q_{a, 0}(Z, S, X) \mid W, S, X, G = \text{O}} \\
  &= \E\bs{\E\bs{q_{a, 0}(Z, S, X) \mid U, W, S, X, G = \text{O}} \mid W, S, X, G = \text{O}} \tag{Iterated expectation} \\
  &= \E\bs{\E\bs{q_{a, 0}(Z, S, X) \mid U, S, X, G = \text{O}} \mid W, S, X, G = \text{O}}. \tag{$W \ind (S, Z) \mid (X, U, G)$}
\end{align*}
On the other hand,
the right-hand side of \cref{eq:surrogate-bridge-function}
can be rewritten as follows:
\begin{align*}
  &\mathrel{\phantom{=}} 
  \frac{f(w, s \mid a, x, \text{E}) f(x \mid \text{E})}{f(w, s, x \mid \text{O})} \\
  &= \int_{\mathcal{U}}
  \frac{f(w, u, s \mid a, x, \text{E}) f(x \mid \text{E})}{f(w, s, x \mid \text{O})}
  \, d u \tag{Total probability} \\
  &= \int_{\mathcal{U}}
  \frac{f(w \mid u, s, a, x, \text{E}) f(u, s \mid a, x, \text{E}) f(x \mid \text{E})}{f(w, s, x \mid \text{O})}
  \, d u \tag{Chain rule} \\
  &= \int_{\mathcal{U}}
  \frac{f(w \mid u, x, \text{E}) f(u, s \mid a, x, \text{E}) f(x \mid \text{E})}{f(w, s, x \mid \text{O})}
  \, d u \tag{$W \ind (A, S) \mid (X, U, G)$} \\
  &= \int_{\mathcal{U}}
  \frac{f(w \mid u, x, \text{O}) f(u, s \mid a, x, \text{E}) f(x \mid \text{E})}{f(w, s, x \mid \text{O})}
  \, d u \tag{$W \ind G \mid (X, U)$} \\
  &= \int_{\mathcal{U}}
  \frac{f(w \mid u, s, x, \text{O}) f(u, s \mid a, x, \text{E}) f(x \mid \text{E})}{f(w, s, x \mid \text{O})}
  \, d u. \tag{$W \ind S \mid (X, U, G)$} \\
  &= \int_{\mathcal{U}}
  \frac{f(u, s \mid a, x, \text{E}) f(x \mid \text{E})}{f(u, s, x \mid \text{O})}
  \cdot
  \frac{f(w \mid u, s, x, \text{O}) f(u, s, x \mid \text{O})}{f(w, s, x \mid \text{O})}
  \, d u \tag{Bayes' rule} \\
  &= \int_{\mathcal{U}}
  \frac{f(u, s \mid a, x, \text{E}) f(x \mid \text{E})}{f(u, s, x \mid \text{O})}
  f(w \mid u, s, x, \text{O})
  \, d u \tag{Simplification} \\
  &= \E\bs{\frac{f(U, S \mid A = a, X, G = \text{E}) f(X \mid G = \text{E})}{f(U, S, X \mid G = \text{O})} \;\middle|\; W = w, S = s, X = x, G = \text{O}}.
\end{align*}
By equating the two expansions above,
we have
\begin{align*}
  &\mathrel{\phantom{=}} \E\bs{\E\bs{q_{a, 0}(Z, S, X) \mid U, S, X, G = \text{O}} \mid W, S, X, G = \text{O}} \\
  &= \E\bs{\frac{f(U, S \mid A = a, X, G = \text{E}) f(X \mid G = \text{E})}{f(U, S, X \mid G = \text{O})} \;\middle|\; W, S, X, G = \text{O}}.
\end{align*}
Applying the completeness condition in \cref{asn:completeness-surrogate},
we obtain
\begin{align*}
  \E\bs{q_{a, 0}(Z, S, X) \mid U, S, X, G = \text{O}}
  = \frac{f(U, S \mid A = a, X, G = \text{E}) f(X \mid G = \text{E})}{f(U, S, X \mid G = \text{O})}.
\end{align*}
This completes the proof of Part 1.

\paragraph{Part 2}

The mean potential outcome $\E\bs{Y(a) \mid G = \text{E}}$
can be identified as follows:
\begin{align}
  &\mathrel{\phantom{=}} \E\bs{q_{a, 0}(Z, S, X) Y \mid G = \text{O}} \notag \\
  &= \E\bs{\E\bs{q_{a, 0}(Z, S, X) Y \mid U, S, X, G = \text{O}} \mid G = \text{O}} \label{eq:sb-proof-1} \\
  &= \E\bs{\E\bs{q_{a, 0}(Z, S, X) \mid U, S, X, G = \text{O}}
  \cdot \E\bs{Y \mid U, S, X, G = \text{O}} \mid G = \text{O}} \label{eq:sb-proof-2} \\
  &= \iiint
  \frac{f(u, s \mid a, x, \text{E}) f(x \mid \text{E})}{f(u, s, x \mid \text{O})}
  \E\bs{Y \mid u, s, x, \text{O}}
  f(u, s, x \mid \text{O}) \, d u \, d s \, d x \label{eq:sb-proof-3} \\
  &= \iiint
  f(u, s \mid a, x, \text{E}) f(x \mid \text{E})
  \E\bs{Y \mid u, s, x, \text{O}}
  \, d u \, d s \, d x \label{eq:sb-proof-4} \\
  &= \iiint
  f(u, s \mid a, x, \text{E}) f(x \mid \text{E})
  \E\bs{Y \mid u, s, x, \text{E}}
  \, d u \, d s \, d x \label{eq:sb-proof-5} \\
  &= \iiint
  f(u, s \mid a, x, \text{E}) f(x \mid \text{E})
  \E\bs{Y \mid u, s, x, a, \text{E}}
  \, d u \, d s \, d x \label{eq:sb-proof-6} \\
  &= \int
  \E\bs{Y \mid A = a, X = x, G = \text{E}} f(x \mid \text{E}) \, d x
  \label{eq:sb-proof-7} \\
  &= \int
  \E\bs{Y(a) \mid X = x, G = \text{E}} f(x \mid \text{E}) \, d x
  \label{eq:sb-proof-8} \\
  &= \E\bs{Y(a) \mid G = \text{E}}, \label{eq:sb-proof-9}
\end{align}
where
\cref{eq:sb-proof-1} applies the law of total expectation,
\cref{eq:sb-proof-2} follows from $Y \ind Z \mid (U, S, X, G = \text{O})$,
\cref{eq:sb-proof-3} plugs in the result from the first part,
\cref{eq:sb-proof-4} simplifies the expression,
\cref{eq:sb-proof-5} follows from $Y \ind G \mid (U, S, X)$,
\cref{eq:sb-proof-6} applies \cref{eq:independence-for-ob-proof},
\cref{eq:sb-proof-7} applies the law of total probability,
\cref{eq:sb-proof-8} uses \cref{asn:unconfoundedness-experimental},
and \cref{eq:sb-proof-9}
follows from the law of total probability.
This completes the proof of Part 2.

\subsection{Proof of \cref{thm:mutiply-robust-identification}} \label{sec:proof-of-thm-mutiply-robust-identification}

We prove this theorem by considering two separate cases:
\begin{enumerate}
  \item the outcome bridge function is correct, or
  \item the surrogate bridge functions are correct.
\end{enumerate}

\paragraph{Case 1: Outcome Bridge is Correct}

Suppose that $h = h_0$ satisfies \cref{eq:outcome-bridge-function}.
Here, $q_1$ and $q_0$ are arbitrary nuisance functions.
Then, the term $\varphi_{\text{O}}(D; \eta)$ 
in \cref{eq:mutiply-robust-identification}
is zero:
\begin{align*}
  &\mathrel{\phantom{=}} \E\bs{\bp{q_1 - q_0} \bp{Y - h_0} \mid G = \text{O}} \\
  &= \E\bs{\E\bs{\bp{q_1 - q_0} \bp{Y - h_0} \;\middle|\; Z, S, X, G = \text{O}} \;\middle|\; G = \text{O}} \tag{Iterated expectation} \\
  &= \E\bs{\bp{q_1 - q_0} \E\bs{Y - h_0 \;\middle|\; Z, S, X, G = \text{O}} \;\middle|\; G = \text{O}} \tag{Pulling out constant} \\
  &= \E\bs{\bp{q_1 - q_0} \cdot 0 \mid G = \text{O}} \tag{$h_0$ satisfies \cref{eq:outcome-bridge-function}} \\
  &= 0.
\end{align*}

We now focus on the remaining terms 
$\varphi_{\text{E}}(D; \eta)$ by considering the two scenarios.

First, suppose that the outcome regression model is correctly specified, i.e., $\bar{h} = \bar{h}_0$,
but the propensity score model $e$ is arbitrary.
Consider the term associated with treatment $A=1$:
\begin{align*}
  &\mathrel{\phantom{=}}
  \E\bs{\frac{\mathbf{1}_{\{G=\text{E}\}}}{\pi_0} \bp{ \frac{A(h_0 - \bar{h}_0(1, X))}{e(X)} + \bar{h}_0(1, X) }} \\
  &= \E\bs{ \frac{A(h_0 - \bar{h}_0(1, X))}{e(X)} + \bar{h}_0(1, X) \;\middle|\; G=\text{E} } \\
  &= \E\bs{ \frac{\E[A(h_0 - \bar{h}_0(1, X)) \mid X, G=\text{E}]}{e(X)} + \bar{h}_0(1, X) \;\middle|\; G=\text{E} }.
\end{align*}
Focusing on the numerator of the first term inside the expectation:
\begin{align*}
  &\mathrel{\phantom{=}}
  \E\bs{A(h_0 - \bar{h}_0(1, X)) \mid X, G=\text{E}} \\
  &= \E\bs{ A \bp{ \E\bs{h_0 \mid A=1, X, G=\text{E}} - \bar{h}_0(1, X) } \mid X, G=\text{E} } \\
  &= \E\bs{ A \bp{ \bar{h}_0(1, X) - \bar{h}_0(1, X) } \mid X, G=\text{E} } \\
  &= 0.
\end{align*}
Thus, the bias correction term vanishes regardless of $e(X)$.
By symmetry, the $A=0$ term also has a zero bias correction term.
By \cref{thm:outcome-bridge-identification} (Part 3), 
$\E[\bar{h}_0(a, X) \mid G=\text{E}] = \E[Y(a)]$.
Therefore, $\E[\varphi_{\text{E}}(D; \eta)] = \tau_0$.

Second, suppose that the propensity score model is correctly specified, i.e., $e = e_0$,
but $\bar{h}$ is arbitrary.
Then, we have
\begin{align*}
  \E\bs{\varphi_{\text{E}}(D; \eta)}
  &= \E\bs{ \frac{A - e_0(X)}{e_0(X)(1 - e_0(X))} (h_0 - \bar{h}) + (\bar{h}(1, X) - \bar{h}(0, X)) \;\middle|\; G=\text{E} }.
\end{align*}
Now,
the coefficient for the arbitrary function $\bar{h}$ is
\begin{align*}
  &\mathrel{\phantom{=}} \E\bs{ -\frac{A - e_0(X)}{e_0(X)(1 - e_0(X))} \bar{h}(A, X) + (\bar{h}(1, X) - \bar{h}(0, X)) \;\middle|\; G=\text{E} } \\
  &= \E\bs{ \E\bs{ -\frac{A}{e_0(X)} \bar{h}(1, X) + \frac{1-A}{1-e_0(X)} \bar{h}(0, X) + \bar{h}(1, X) - \bar{h}(0, X) \;\middle|\; X, G=\text{E} } \;\middle|\; G=\text{E} } \\
  &= \E\bs{ -\frac{e_0(X)}{e_0(X)} \bar{h}(1, X) + \frac{1-e_0(X)}{1-e_0(X)} \bar{h}(0, X) + \bar{h}(1, X) - \bar{h}(0, X) \;\middle|\; G=\text{E} } \\
  &= 0.
\end{align*}
The arbitrary $\bar{h}$ cancels out. 
The remaining term involving $h_0$ can be simplified as follows:
\begin{align*}
  &\mathrel{\phantom{=}}
  \E\bs{ \frac{A - e_0(X)}{e_0(X)(1 - e_0(X))} h_0 \;\middle|\; G=\text{E} } \\
  &= \E\bs{ \frac{A h_0}{e_0(X)} - \frac{(1-A)h_0}{1-e_0(X)} \;\middle|\; G=\text{E} } \\
  &= \E\bs{ \E[h_0 \mid A=1, X, G = \text{E}] - \E[h_0 \mid A=0, X, G = \text{E}] \;\middle|\; G=\text{E} } \\
  &= \tau_0.
\end{align*}

In summary,
under $h=h_0$, 
if either $\bar{h}=\bar{h}_0$ or $e=e_0$, 
we have $\E\bs{\varphi(D; \eta)} = \tau_0$.

\paragraph{Case 2: Surrogate Bridges are Correct}

Assume that the surrogate bridges are correctly specified, 
i.e., $q_a = q_{a,0}$ satisfies \cref{eq:surrogate-bridge-function}.
Let $h$ be an arbitrary outcome bridge function.
Let $\tau_h$ denote the difference-in-means
on this arbitrary function $h$, i.e.,
\begin{align*}
  \tau_h \equiv \E[h(W, S, X) \mid A = 1, G=\text{E}] - \E[h(W, S, X) \mid A = 0, G=\text{E}].
\end{align*}

We decompose $\varphi_{\text{O}}(D; \eta)$ into parts involving $Y$ and $h$:
\begin{align*}
  \E\bs{\varphi_{\text{O}}(D; \eta)}
  &= \E\bs{ (q_{1,0} - q_{0,0})Y \mid G=\text{O} } - \E\bs{ (q_{1,0} - q_{0,0})h \mid G=\text{O} }.
\end{align*}
Using \cref{thm:surrogate-bridge-identification}, the first term identifies the target ATE:
\begin{align*}
  \E\bs{ (q_{1,0} - q_{0,0})Y \mid G=\text{O} } = \tau_0.
\end{align*}
For the second term, since $q_{a,0}$ satisfies \cref{eq:surrogate-bridge-function}, 
it reweights the observational distribution of any variable (including $h$) 
to match the experimental distribution.
Thus, we have
\begin{align*}
  \E\bs{ q_{a,0} h \mid G=\text{O} } = \E\bs{ h(W, S, X) \mid A=a, G=\text{E} }
\end{align*}
Therefore, the observational component simplifies to:
\begin{align*}
  \E\bs{\varphi_{\text{O}}(D; \eta)}
  = \tau_0 - \tau_h.
\end{align*}

Now, we analyze the experimental component $\varphi_{\text{E}}(D; \eta)$
by considering the two scenarios.

First, suppose that the propensity score model is correctly specified, i.e., $e = e_0$,
but $\bar{h}$ is arbitrary.
Then, we have
\begin{align*}
  \E\bs{\varphi_{\text{E}}(D; \eta)}
  = \E\bs{ \frac{A - e_0(X)}{e_0(X)(1 - e_0(X))} (h - \bar{h}) + (\bar{h}(1, X) - \bar{h}(0, X)) \;\middle|\; G=\text{E} }.
\end{align*}
As shown in Case 1,
the coefficient for the arbitrary function $\bar{h}$ is zero:
\begin{align*}
  \E\bs{ -\frac{A - e_0(X)}{e_0(X)(1 - e_0(X))} \bar{h}(A, X) + (\bar{h}(1, X) - \bar{h}(0, X)) \;\middle|\; G=\text{E} }
  = 0.
\end{align*}
Thus, the arbitrary $\bar{h}$ cancels out.
The remaining term involving $h$ can be simplified as follows:
\begin{align*}
  &\mathrel{\phantom{=}}
  \E\bs{ \frac{A - e_0(X)}{e_0(X)(1 - e_0(X))} h \;\middle|\; G=\text{E} } \\
  &= \E\bs{ \frac{A h}{e_0(X)} - \frac{(1-A)h}{1-e_0(X)} \;\middle|\; G=\text{E} } \\
  &= \E\bs{ \E[h \mid A=1, X, G = \text{E}] - \E[h \mid A=0, X, G = \text{E}] \;\middle|\; G=\text{E} } \\
  &= \tau_h.
\end{align*}
Thus, under $q_a = q_{a,0}$ and $e = e_0$,
we have
\begin{align*}
  \E\bs{\varphi(D; \eta)}
  = \E\bs{\varphi_{\text{O}}(D; \eta)} + \E\bs{\varphi_{\text{E}}(D; \eta)}
  = (\tau_0 - \tau_h) + \tau_h 
  = \tau_0.
\end{align*}

Second, suppose that $\bar{h} = \bar{h}^{(h)}$
($h$ is potentially misspecified)
and that the propensity score model $e$ is arbitrary.
Following the derivation in Case 1,
by replacing $h_0$ with $h$ and $\bar{h}_0$ with this specific $\bar{h}$,
we see that the bias correction term vanishes:
\begin{align*}
  \E\bs{ \frac{A(h - \bar{h}(1, X))}{e(X)} \;\middle|\; X, G=\text{E} }
  = \frac{1}{e(X)} \underbrace{\E\bs{ A(h - \E[h \mid 1, X, G = \text{E}]) \mid X, G=\text{E} }}_{=0} = 0.
\end{align*}
By symmetry, the $A=0$ term also has a zero bias correction term.
Therefore, we have
\begin{align*}
  \E\bs{\varphi_{\text{E}}(D; \eta)}
  = \E[\bar{h}(1, X) - \bar{h}(0, X) \mid G=\text{E}]
  = \tau_h.
\end{align*}
Thus, under $q_a = q_{a,0}$ and $\bar{h} = \E[h \mid A, X, G=\text{E}]$,
we have
\begin{align*}
  \E\bs{\varphi(D; \eta)}
  = \E\bs{\varphi_{\text{O}}(D; \eta)} + \E\bs{\varphi_{\text{E}}(D; \eta)}
  = (\tau_0 - \tau_h) + \tau_h 
  = \tau_0.
\end{align*}

In summary,
under $q_a = q_{a,0}$ for $a \in \{0, 1\}$,
if either $e = e_0$ or $\bar{h} = \E[h \mid A, X, G=\text{E}]$,
we have $\E\bs{\varphi(D; \eta)} = \tau_0$.

\subsection{Proof of \cref{thm:consistency}} \label{sec:proof-of-thm-consistency}

We establish consistency under four different scenarios 
corresponding to the multiply robust identification conditions.
Each part follows the same general approach:
\begin{enumerate}
  \item Define an \emph{oracle estimator} that uses the correctly specified nuisance functions in place of the estimated
ones.
  \item Show the oracle estimator is consistent by the weak law of large numbers (WLLN).
  \item Bound the difference between the estimated and oracle estimators 
    using the $L_2$ convergence rates from \cref{asn:convergence-rates}.
\end{enumerate}
The key insight is that under any of the four identification conditions,
the oracle estimator has zero bias,
and the estimation error vanishes as the nuisance estimators converge to their limits.

\paragraph{Part 1}

Suppose that $\tilde{h} = h_0$ and $\tilde{\bar{h}}^{(h_0)} = \bar{h}_0$.
Consider a fixed treatment level $a \in \{0, 1\}$.
By the WLLN, the oracle estimator 
for the mean potential outcome under treatment $a$ is consistent:
\begin{align*}
  \frac{1}{N_{\text{E}}} \sum_{k = 1}^K \sum_{i \in \mathcal{I}_{\text{E},k}} \bar{h}_0(a, X_i)
  \xrightarrow{p}
  \E\bs{\bar{h}_0(a, X) \mid G = \text{E}}.
\end{align*}
Then, the difference between the estimated component for treatment $a$ and its oracle counterpart is
\begin{align*}
  \frac{1}{K} \sum_{k = 1}^K \Delta^{\text{OB-OR}}_{a,k},
\end{align*}
where
\begin{align*}
  \Delta^{\text{OB-OR}}_{a,k} \equiv \frac{1}{|\mathcal{I}_{\text{E},k}|} \sum_{i \in \mathcal{I}_{\text{E},k}}
  \bp{\hat{\bar{h}}_k(a, X_i) - \bar{h}_0(a, X_i)}.
\end{align*}
We analyze the convergence of $\Delta^{\text{OB-OR}}_{a,k}$ for a fixed $k$.
Conditioning on the training data $\mathcal{D}{-k}$,
we decompose $\Delta^{\text{OB-OR}}{a,k}$ into a bias term and a
stochastic term:
\begin{align*}
  \Delta^{\text{OB-OR}}_{a,k} = \underbrace{\E\bs{\Delta^{\text{OB-OR}}_{a,k} \mid \mathcal{D}_{-k}}}_{\text{Bias Term}} +
\underbrace{\bp{\Delta^{\text{OB-OR}}_{a,k} - \E\bs{\Delta^{\text{OB-OR}}_{a,k} \mid \mathcal{D}_{-k}}}}_{\text{Stochastic Term}}.
\end{align*}
For the bias term, applying Jensen's inequality yields
\begin{align*}
  \abs{\E\bs{\Delta^{\text{OB-OR}}_{a,k} \mid \mathcal{D}_{-k}}}
  &= \abs{\E\bs{\hat{\bar{h}}_k(a, X) - \bar{h}_0(a, X) \;\middle|\; \mathcal{D}_{-k}}} \\
  &\leq \E\bs{\abs{\hat{\bar{h}}_k(a, X) - \bar{h}_0(a, X)} \;\middle|\; \mathcal{D}_{-k}} \\
  &\leq \|\hat{\bar{h}}_k(a, \cdot) - \bar{h}_0(a, \cdot)\|_2.
\end{align*}
To bound this norm, we use the triangle inequality and Lipschitz continuity:
\begin{align*}
  &\mathrel{\phantom{=}}
  \|\hat{\bar{h}}_k(a, \cdot) - \bar{h}_0(a, \cdot)\|_2 \\
  &\leq \|\hat{\bar{h}}_k(a, \cdot) - \tilde{\bar{h}}^{(\hat{h}_k)}(a, \cdot)\|_2 + \|\tilde{\bar{h}}^{(\hat{h}_k)}(a, \cdot) -
  \tilde{\bar{h}}^{(h_0)}(a, \cdot)\|_2 + \|\tilde{\bar{h}}^{(h_0)}(a, \cdot) - \bar{h}_0(a, \cdot)\|_2 \\
  &\precsim \|\hat{\bar{h}}_k - \tilde{\bar{h}}^{(\hat{h}_k)}\|_2 + \|\hat{h}_k - h_0\|_2,
\end{align*}
where the second inequality follows from $\|\tilde{\bar{h}}^{(g_1)} - \tilde{\bar{h}}^{(g_2)}\|_2 \leq C \|g_1 - g_2\|_2$
for some constant $C > 0$
and the fact that $\tilde{\bar{h}}^{(h_0)} = \bar{h}_0$ by assumption.
By \cref{asn:convergence-rates}, the first term is $O_P(\delta_{\bar{h}, N})$ and the second term is $O_P(\delta_{h, N})$.
Thus, the bias term is $O_P(\delta_{\bar{h}, N} + \delta_{h, N})$.
For the stochastic term, we apply Chebyshev's inequality.
Conditional on $\mathcal{D}_{-k}$, the variance is bounded by the squared $L_2$ error:
\begin{align*}
  \Var\bs{\Delta^{\text{OB-OR}}_{a,k} \mid \mathcal{D}_{-k}}
  &= \frac{1}{|\mathcal{I}_{\text{E},k}|} \Var\bs{\hat{\bar{h}}_k(a, X) - \bar{h}_0(a, X) \;\middle|\; \mathcal{D}_{-k}} \\
  &\leq \frac{1}{|\mathcal{I}_{\text{E},k}|} \E\bs{\bp{\hat{\bar{h}}_k(a, X) - \bar{h}_0(a, X)}^2 \;\middle|\; \mathcal{D}_{-k}} \\
  &= \frac{\|\hat{\bar{h}}_k(a, \cdot) - \bar{h}_0(a, \cdot)\|_2^2}{|\mathcal{I}_{\text{E},k}|}.
\end{align*}
Thus, the stochastic term is $O_P((\delta_{\bar{h}, N} + \delta_{h, N}) / \sqrt{N})$.
Combining these results, we have
\begin{align*}
  \Delta^{\text{OB-OR}}_{a,k} = O_P(\delta_{\bar{h}, N} + \delta_{h, N}).
\end{align*}
Since $\delta_{\bar{h}, N} \to 0$ and $\delta_{h, N} \to 0$, the difference converges to zero in probability. It follows that
$\hat{\tau}_{\text{OB-OR}} \xrightarrow{p} \tau_0$.

\paragraph{Part 2}

Suppose that $\tilde{h} = h_0$ and $\tilde{e} = e_0$.
The consistency of the oracle estimator follows from the WLLN.
The difference between the estimated and oracle estimators for the $A=1$ term is the average of
\[
  \Delta^{\text{OB-IPW}}_{1, k} \equiv \frac{1}{|\mathcal{I}_{\text{E},k}|} \sum_{i \in \mathcal{I}_{\text{E},k}} A_i
  \bp{\frac{\hat{h}_k}{\hat{e}_k} - \frac{h_0}{e_0}}.
\]
Following the same bias-variance decomposition strategy as in Part 1,
and using the identity 
\[
  \frac{\hat{a}}{\hat{b}} - \frac{a}{b} = \frac{1}{\hat{b}}(\hat{a}-a) - \frac{a}{\hat{b}b}(\hat{b}-b),
\]
we can bound the bias term by $O_P(\|\hat{h}_k - h_0\|_2 + \|\hat{e}_k - e_0\|_2)$
and the stochastic term by $O_P((\|\hat{h}_k - h_0\|_2 + \|\hat{e}_k - e_0\|_2)/\sqrt{N})$.
Given \cref{asn:convergence-rates}, the total error is $O_P(\delta_{h,N} + \delta_{e,N})$, which vanishes asymptotically.
The same applies to the $A=0$ term. Thus, $\hat{\tau}_{\text{OB-IPW}} \xrightarrow{p} \tau_0$.

\paragraph{Part 3}

Suppose that $\tilde{q}_a = q_{a,0}$ for $a \in \{0, 1\}$.
By arguments analogous to previous parts,
the estimation error for the estimator is controlled by the convergence rate of the surrogate bridge functions.
Specifically, the error is of order $O_P(\delta_{q,N})$.
Since $\delta_{q,N} \to 0$ (\cref{asn:convergence-rates}),
it follows that $\hat{\tau}_{\text{SB}} \xrightarrow{p} \tau_0$.

\paragraph{Part 4}

Suppose that the limit functions $\tilde{\eta} = (\tilde{e}, \tilde{h}, \tilde{\bar{h}}, \tilde{q}_0, \tilde{q}_1)$ satisfy at
least one of the identification conditions specified in \cref{thm:mutiply-robust-identification}.
Let $\tilde{\tau}_{\text{MR}}$ be the oracle estimator constructed using these limit functions.
By \cref{thm:mutiply-robust-identification} and the WLLN, $\tilde{\tau}_{\text{MR}}$ is consistent for $\tau_0$.
The difference between the feasible estimator and the oracle is given by
\begin{align*}
  \hat{\tau}_{\text{MR}} - \tilde{\tau}_{\text{MR}} = \frac{1}{K} \sum_{k=1}^K \Delta^{\text{MR}}_{k}, \quad \text{where} \quad
  \Delta^{\text{MR}}_{k} \equiv \frac{1}{|\mathcal{I}_{k}|} \sum_{i \in \mathcal{I}_{k}} \bp{\varphi(D_i; \hat{\eta}_k) -
\varphi(D_i; \tilde{\eta})}.
\end{align*}
The influence function $\varphi$ consists of linear combinations and products of the nuisance functions.
Using the algebraic identity $\hat{a}\hat{b} - \tilde{a}\tilde{b} = (\hat{a} - \tilde{a})\hat{b} + \tilde{a}(\hat{b} -
\tilde{b})$ and assuming the propensity scores are bounded away from zero and one, the difference $\varphi(D_i; \hat{\eta}_k) -
\varphi(D_i; \tilde{\eta})$ can be bounded by a linear combination of the absolute errors of individual nuisance components.
We again apply the bias-variance decomposition conditioning on $\mathcal{D}_{-k}$.
The bias term is bounded by the sum of the $L_2$ norms of the estimation errors:
\begin{align*}
  \abs{\E\bs{\Delta^{\text{MR}}_{k} \mid \mathcal{D}_{-k}}}
  \precsim \|\hat{e}_k - \tilde{e}\|_2 + \|\hat{h}_k - \tilde{h}\|_2 + \|\hat{\bar{h}}_k - \tilde{\bar{h}}\|_2 + \sum_{a}
\|\hat{q}_{a,k} - \tilde{q}_a\|_2.
\end{align*}
Similarly, by Chebyshev's inequality, the stochastic term is of the order $O_P((\delta_{e,N} + \delta_{h,N} +
\delta_{\bar{h},N} + \delta_{q,N})/\sqrt{N})$.
Combining these, the total estimation error is $O_P(\delta_{e,N} + \delta_{h,N} + \delta_{\bar{h},N} + \delta_{q,N})$.
Under \cref{asn:convergence-rates}, all these terms converge to zero, implying $\hat{\tau}_{\text{MR}} \xrightarrow{p} \tau_0$.

\subsection{Proof of \cref{thm:asymptotic-normality}} \label{sec:proof-of-thm-asymptotic-normality}

Define $\psi(D; \eta)$ as
\begin{align*}
  \psi(D; \eta) \equiv \frac{\mathbf{1}_{\{G=\text{E}\}}}{\pi_0} \psi_{\text{E}}(D; \eta) +
\frac{\mathbf{1}_{\{G=\text{O}\}}}{1 - \pi_0} \psi_{\text{O}}(D; \eta),
\end{align*}
where 
\begin{align*}
  \psi_{\text{E}}(D; \eta) &\equiv \frac{(A - e(X)) (h(W, S, X) - \bar{h}(A, X))}{e(X)(1 - e(X))} + \bar{h}(1, X) - \bar{h}(0,
X) - \tau_0, \\
  \psi_{\text{O}}(D; \eta) &\equiv (q_{1}(Z, S, X) - q_{0}(Z, S, X)) (Y - h(W, S, X)).
\end{align*}
Note that $\hat{\tau}_{\text{MR}}$ satisfies 
$\hat{\tau}_{\text{MR}} - \tau_0 = \frac{1}{N} \sum_{k=1}^K \sum_{i \in \mathcal{I}_k} \psi(D_i; \hat{\eta}_k) + o_p(N^{-1/2})$.
We decompose the scaled difference into a leading oracle term and a remainder term:
\begin{align*}
  \sqrt{N}(\hat{\tau}_{\text{MR}} - \tau_0)
  = \underbrace{\frac{1}{\sqrt{N}} \sum_{i=1}^N \psi(D_i; \eta_0)}_{\text{(I) Main Term}}
   + \underbrace{\frac{1}{\sqrt{N}} \sum_{k=1}^K \sum_{i \in \mathcal{I}_k} \bp{\psi(D_i; \hat{\eta}_k) - \psi(D_i;
     \eta_0)}}_{\text{(II) Remainder Term}}.
\end{align*}

Since $\eta_0$ consists of the true functions, 
$\E[\psi(D; \eta_0)] = 0$. 
By the central limit theorem (CLT),
term (I) converges in distribution
to $\mathcal{N}(0, V)$, where $V = \Var(\psi(D; \eta_0))$. 
It remains to show that term (II) is $o_p(1)$.

For each fold $k$, define the average estimation error on the validation set as:
\begin{align*}
  \Delta_k^{\text{MR}} \equiv \frac{1}{|\mathcal{I}_k|} \sum_{i\in \mathcal{I}_k} \bp{\psi(D_i; \hat{\eta}_k) - \psi(D_i;
\eta_0)}.
\end{align*}
Conditioning on the training data $\mathcal{D}_{-k}$, 
we decompose $\Delta_k^{\text{MR}}$ into a bias term $B_k \equiv
\E[\Delta_k^{\text{MR}} \mid \mathcal{D}_{-k}]$ and a stochastic deviation term $(\Delta_k^{\text{MR}} - B_k)$.

First, for the stochastic deviation term, the conditional variance is bounded by the squared estimation errors:
\begin{align*}
  \Var\bs{\sqrt{N}(\Delta_k^{\text{MR}} - B_k) \;\middle|\; \mathcal{D}_{-k}}
  &\asymp \E\bs{\bp{\psi(D; \hat{\eta}_k) - \psi(D; \eta_0)}^2 \;\middle|\; \mathcal{D}_{-k}} \\
  &\precsim \delta_{e,N}^2 + \delta_{h,N}^2 + \delta_{\bar{h},N}^2 + \delta_{q,N}^2.
\end{align*}
Under \cref{asn:convergence-rates}, this vanishes, 
so $\sqrt{N}(\Delta_k^{\text{MR}} - B_k) = o_p(1)$.

Next, we analyze the bias term $B_k$. 
Using the independence of samples across groups, we can write $B_k = B_{\text{E}, k} +
B_{\text{O}, k}$, where $B_{g, k}$ corresponds to the bias from group $g$.

Consider the observational bias $B_{\text{O}, k}$ first.
Since the true outcome bridge function satisfies $\E[Y - h_0 \mid Z, S, X, G=\text{O}] = 0$, the expectation of the true
influence function is zero.
We can thus rewrite the bias as:
\begin{align*}
  B_{\text{O}, k}
  &= \E\bs{\psi_{\text{O}}(D; \hat{\eta}_k) - \psi_{\text{O}}(D; \eta_0) \;\middle|\; G=\text{O}, \mathcal{D}_{-k}} \\
  &= \E\bs{(\hat{q}_{1,k} - \hat{q}_{0,k})(Y - \hat{h}_k) \;\middle|\; G=\text{O}, \mathcal{D}_{-k}}.
\end{align*}
We decompose this term by adding and subtracting the true treatment difference $(q_{1,0} - q_{0,0})$ interacting with the
outcome bridge error $(\hat{h}_k - h_0)$:
\begin{align*}
  B_{\text{O}, k}
  &= -\underbrace{\E\bs{(q_{1,0} - q_{0,0})(\hat{h}_k - h_0) \;\middle|\; G=\text{O},
\mathcal{D}_{-k}}}_{L_{\text{O},k}} \\
  &\quad - \underbrace{\E\bs{\bp{(\hat{q}_{1,k} - q_{1,0}) - (\hat{q}_{0,k} - q_{0,0})} (\hat{h}_k - h_0)
\;\middle|\; G=\text{O}, \mathcal{D}_{-k}}}_{R_{\text{O}, k}}.
\end{align*}
where $L_{\text{O}, k}$ is the leading term
and $R_{\text{O}, k}$ is the remainder term.

The remainder term $R_{\text{O}, k}$ 
involves the product of estimation errors from both the surrogate and outcome bridge functions.
Note that since bridge estimation solves conditional moment equations,
\[
  T(h) = \E[Y \mid Z,S,X,G=\text{O}], \quad
  T^*(q_a) = \E[h \mid A=a,S,X,G=\text{E}],
\]
we can exploit the operators $T$ and $T^*$
to relate the errors in different metrics.
Specifically, we can bound it using the Cauchy-Schwarz inequality in two alternative ways.
First, applying the law of iterated expectations and the definition of the operator $T$,
we have
\begin{align*}
  &\mathrel{\phantom{=}}
  \abs{\E\bs{(\hat{q}_{a,k} - q_{a,0}) (\hat{h}_k - h_0) \;\middle|\; G=\text{O}, \mathcal{D}_{-k}}} \\
  &= \abs{\E\bs{(\hat{q}_{a,k} - q_{a,0}) \E\bs{\hat{h}_k - h_0 \mid Z, S, X, G=\text{O}, \mathcal{D}_{-k}} \;\middle|\; G=\text{O}, \mathcal{D}_{-k}}} \\
  &= \abs{\E\bs{(\hat{q}_{a,k} - q_{a,0}) T(\hat{h}_k - h_0)(Z,S,X) \;\middle|\; G=\text{O}, \mathcal{D}_{-k}}} \\
  &\leq \|\hat{q}_{a,k} - q_{a,0}\|_2 \|T(\hat{h}_k - h_0)\|_2.
\end{align*}
Similarly, using the adjoint operator $T^*$, we have
\begin{align*}
  \abs{\E\bs{(\hat{q}_{a,k} - q_{a,0}) (\hat{h}_k - h_0) \;\middle|\; G=\text{O}, \mathcal{D}_{-k}}}
  \leq \|T^*(\hat{q}_{a,k} - q_{a,0})\|_2 \|\hat{h}_k - h_0\|_2.
\end{align*}
Combining these bounds and taking the minimum over the two inequalities, we obtain
\begin{align*}
  \abs{R_{\text{O}, k}}
  &\leq \sum_{a \in \{0,1\}} \min \bc{ \|\hat{q}_{a,k} - q_{a,0}\|_2 \|T(\hat{h}_k - h_0)\|_2, \;
\|T^*(\hat{q}_{a,k} - q_{a,0})\|_2 \|\hat{h}_k - h_0\|_2 }.
\end{align*}
By the definitions in \cref{asn:convergence-rates},
the first term in the minimum equals $\delta_{q,N} \rho_{h,N}$
and the second equals $\rho_{q,N} \delta_{h,N}$.
By \cref{asn:product-rates},
$\min\{\delta_{q,N} \rho_{h,N}, \rho_{q,N} \delta_{h,N}\} = o_p(N^{-1/2})$.
Thus, $\abs{R_{\text{O}, k}} = o_p(N^{-1/2})$.

Now consider the experimental bias $B_{\text{E}, k}$.
The component $\psi_{\text{E}}(D; \hat{\eta}_k)$ corresponds to 
an augemented inverse probability (AIPW) estimator where the outcome variable is the
estimated bridge function $\hat{h}_k(W, S, X)$.
Therefore, the bias depends on the error of $\hat{\bar{h}}_k$ in approximating the conditional mean of this pseudo-outcome,
denoted by $\bar{h}^{(\hat{h}_k)}(A, X) \equiv \E[\hat{h}_k \mid A, X, G=\text{E}]$.
Using the standard double robustness argument, we decompose the bias as:
\begin{align*}
  B_{\text{E}, k}
  &= \E\bs{\psi_{\text{E}}(D; \hat{\eta}_k) \;\middle|\; G=\text{E}, \mathcal{D}_{-k}} \\
  &= \underbrace{\E\bs{\bar{h}^{(\hat{h}_k)}(1,X) - \bar{h}^{(\hat{h}_k)}(0,X) - \tau_0 \;\middle|\; G=\text{E},
\mathcal{D}_{-k}}}_{L_{\text{E}, k}} + R_{\text{E}, k},
\end{align*}
where $L_{\text{E}, k}$ is the leading term if the propensity score were known, and
the remainder $R_{\text{E}, k}$ captures the second-order error from the AIPW procedure:
\begin{align*}
  R_{\text{E}, k}
  &= \E\bs{ \frac{e_0(X) - \hat{e}_k(X)}{\hat{e}_k(X)} \bp{\bar{h}^{(\hat{h}_k)}(1, X) - \hat{\bar{h}}_k(1, X)} \;\middle|\;
G=\text{E}, \mathcal{D}_{-k} } \\
  &\quad + \E\bs{ \frac{\hat{e}_k(X) - e_0(X)}{1 - \hat{e}_k(X)} \bp{\bar{h}^{(\hat{h}_k)}(0, X) - \hat{\bar{h}}_k(0, X)}
\;\middle|\; G=\text{E}, \mathcal{D}_{-k} }.
\end{align*}
Assuming the propensity scores are bounded away from zero and one, and applying the Cauchy-Schwarz inequality, we have:
\begin{align*}
  \abs{R_{\text{E}, k}} \precsim \|\hat{e}_k - e_0\|_2 \bp{ \|\hat{\bar{h}}_k(1, \cdot) - \bar{h}^{(\hat{h}_k)}(1, \cdot)\|_2 +
\|\hat{\bar{h}}_k(0, \cdot) - \bar{h}^{(\hat{h}_k)}(0, \cdot)\|_2 } = O_p(\delta_{e, N} \delta_{\bar{h}, N}).
\end{align*}
By \cref{asn:product-rates}, this term is $o_p(N^{-1/2})$.

For the leading term $L_{\text{E}, k}$, we can rewrite it using the conditional expectation:
\begin{align*}
  L_{\text{E}, k}
  &= \E\bs{\E[\hat{h}_k \mid A=1, X, G=\text{E}] - \E[\hat{h}_k \mid A=0, X, G=\text{E}] \;\middle|\; G=\text{E}} - \tau_0 \\
  &= \sum_{a \in \{0,1\}} (-1)^{1-a} \E\bs{\hat{h}_k - h_0 \mid A=a, G=\text{E}},
\end{align*}
since $\tau_0 = \E[h_0 \mid A=1, G=\text{E}] - \E[h_0 \mid A=0, G=\text{E}]$.
Using the identification property of the surrogate bridge functions $q_{a, 0}$, we have:
\begin{align*}
  \E\bs{\hat{h}_k - h_0 \mid A=a, G=\text{E}}
  &= \E\bs{q_{a,0}(Z,S,X) (\hat{h}_k - h_0) \mid G=\text{O}}.
\end{align*}
Substituting this back into $L_{\text{E}, k}$, we find that it exactly cancels the leading term $L_{\text{O}, k}$ from the
observational bias.
Thus, the total bias $B_k = B_{\text{E}, k} + B_{\text{O}, k}$ is the sum of $o_p(N^{-1/2})$ remainder terms. This confirms
that term (II) is negligible, completing the proof.

\subsection{Proof of \cref{thm:consistency-of-variance-estimator}} \label{sec:proof-of-thm-consistency-of-variance-estimator}

Let $\tilde{V}$ be the oracle variance estimator, constructed by replacing the estimated nuisance functions and the estimated
ATE with their true counterparts:
\begin{align*}
  \tilde{V}
  &\equiv \frac{N}{N_{\text{E}}^2} \sum_{k=1}^K \sum_{i \in \mathcal{I}_{\text{E}, k}} \psi_{\text{E}, i, 0}^2
   + \frac{N}{N_{\text{O}}^2} \sum_{k=1}^K \sum_{i \in \mathcal{I}_{\text{O}, k}} \psi_{\text{O}, i, 0}^2,
\end{align*}
where
\begin{align*}
  \psi_{\text{E}, i, 0} &\equiv \frac{(A_i - e_0(X_i)) (h_0(W_i, S_i, X_i) - \bar{h}_0(A_i, X_i))}{e_0(X_i)(1 - e_0(X_i))} +
\bar{h}_0(1, X_i) - \bar{h}_0(0, X_i) - \tau_0, \\
  \psi_{\text{O}, i, 0} &\equiv (q_{1, 0}(Z_i, S_i, X_i) - q_{0, 0}(Z_i, S_i, X_i)) (Y_i - h_0(W_i, S_i, X_i)).
\end{align*}
Similarly, let $\psi_{\text{E}, i, k}$ and $\psi_{\text{O}, i, k}$ 
denote the components evaluated at the estimated parameters
$\hat{\eta}_k$ and $\hat{\tau}_{\text{MR}}$ for fold $k$. 
Since $N_{\text{E}} / N \xrightarrow{p} \pi_0$ and $N_{\text{O}} / N \xrightarrow{p} 1 - \pi_0$, the WLLN implies
$\tilde{V} \xrightarrow{p} V$. Thus, it suffices to show $\abs{\hat{V} - \tilde{V}} = o_p(1)$.

We first analyze the observational component. Using the inequality $\abs{a^2 - b^2} \leq (a - b)^2 + 2\abs{b}\abs{a - b}$, we
bound the difference on each fold $k$:
\begin{align*}
  &\mathrel{\phantom{=}}
  \abs{\frac{1}{N_{\text{O}}} \sum_{k=1}^K \sum_{i \in \mathcal{I}_{\text{O}, k}} (\psi_{\text{O}, i, k}^2 - \psi_{\text{O}, i, 0}^2)} \\
  &\leq \frac{1}{N_{\text{O}}} \sum_{k=1}^K \sum_{i \in \mathcal{I}_{\text{O}, k}} (\psi_{\text{O}, i, k} - \psi_{\text{O}, i, 0})^2
   + 2 \sqrt{\frac{1}{N_{\text{O}}} \sum_{i \in \mathcal{I}_{\text{O}}} \psi_{\text{O}, i, 0}^2} \sqrt{\frac{1}{N_{\text{O}}}
     \sum_{k=1}^K \sum_{i \in \mathcal{I}_{\text{O}, k}} (\psi_{\text{O}, i, k} - \psi_{\text{O}, i, 0})^2}.
\end{align*}
The first term under the square root is $O_p(1)$ as it converges to the finite variance of the observational influence
function. For the mean squared error term, by the boundedness of the data and nuisance functions, and applying the
Cauchy-Schwarz inequality, we have:
\begin{align*}
  \frac{1}{|\mathcal{I}_{\text{O}, k}|} \sum_{i \in \mathcal{I}_{\text{O}, k}} (\psi_{\text{O}, i, k} - \psi_{\text{O}, i, 0})^2
  &\precsim \frac{1}{|\mathcal{I}_{\text{O}, k}|} \sum_{i \in \mathcal{I}_{\text{O}, k}} \bp{ (\hat{q}_{1, k} - q_{1, 0})^2 + (\hat{q}_{0, k} - q_{0, 0})^2 + (\hat{h}_k - h_0)^2 } \\
  &= \|\hat{q}_{1, k} - q_{1, 0}\|{2}^2 + \|\hat{q}_{0, k} - q_{0, 0}\|_{2}^2 + \|\hat{h}_k - h_0\|_{2}^2 + o_p(1),
\end{align*}
Under \cref{asn:convergence-rates}, these norms vanish, ensuring the observational component is $o_p(1)$.

Similarly, for the experimental component, the squared difference $(\psi_{\text{E}, i, k} - \psi_{\text{E}, i, 0})^2$ is
controlled by the estimation errors of $e$, $h$, $\bar{h}$, and the estimated ATE:
\begin{align*}
  \frac{1}{|\mathcal{I}_{\text{E}, k}|} \sum_{i \in \mathcal{I}_{\text{E}, k}} (\psi_{\text{E}, i, k} - \psi_{\text{E}, i, 0})^2
  &\precsim \|\hat{e}_k - e_0\|_2^2 + \|\hat{h}_k - h_0\|_2^2 + \|\hat{\bar{h}}_k - \bar{h}_0\|2^2 + (\hat{\tau}_{\text{MR}} - \tau_0)^2 + o_p(1).
\end{align*}
By \cref{asn:convergence-rates} and \cref{thm:asymptotic-normality}, all terms on the right-hand side are $o_p(1)$. Combining
the results, we conclude $\abs{\hat{V} - \tilde{V}} = o_p(1)$.

\subsection{Proof of \cref{prp:surjectivity-uniqueness}} \label{sec:proof-of-surjectivity-uniqueness}

By the fundamental theorem of linear operators,
the kernel of an operator is the orthogonal complement of the range of its adjoint.
In particular,
\begin{align*}
  \Ker(T) = \bp{\Range(T^*)}^\perp, \quad
  \Ker(T^*) = \bp{\Range(T)}^\perp.
\end{align*}

Since $T$ is surjective,
$\Range(T) = L_2\bp{Z, S, X}$.
Thus, $\Ker(T^*) = \bp{\Range(T)}^\perp = \{0\}$,
which implies that $T^*$ is injective.
Consequently,
if there exist two functions
$q_{a, 1}$ and $q_{a, 2}$
satisfying \cref{eq:surrogate-bridge-function},
then $T^*(q_{a, 1} - q_{a, 2}) = 0$.
By the injectivity of $T^*$,
we have $q_{a, 1} = q_{a, 2}$ almost surely.
The proof for the injectivity of $T$ 
and the uniqueness of the outcome bridge function
is similar.

\subsection{Proof of \cref{thm:eif}} \label{sec:proof-of-thm-eif}

The proof consists of five main steps.
First, we introduce the setup and notation.
Second, we factorize the density of the observed data.
Third, we characterize the tangent space of the distribution of the observed data.
Fourth, we derive the influence function
using the pathwise derivative approach.
Finally, we verify the conjectured influence function
is indeed the EIF, which gives the semiparametric efficiency bound.

\paragraph{Setup and Notation}

Let $\mathcal{P}$
be the collection of all possible distributions of the observed data
\[
  D = \bp{\mathbf{1}_{\{G = \text{O}\}} Y, W, \mathbf{1}_{\{G = \text{O}\}} Z, S, \mathbf{1}_{\{G = \text{E}\}} A, X, G},
\]
where for each distribution in $\mathcal{P}$,
there exists a solution $h_0$ to the outcome bridge equation \cref{eq:outcome-bridge-function}.
Consider a regular parametric submodel
$\{P_t: t \in (-\varepsilon, \varepsilon)\} \subset \mathcal{P}$
with density $f_t(D)$
such that $P_0$ equals the true distribution $P$ of $D$.
For any $t$-indexed functional $g_t$,
we denote its pathwise derivative at $t = 0$ by
\begin{align*}
  \partial_t g_t \equiv \left. \frac{\partial}{\partial t} g_t(\cdot) \right|_{t = 0}.
\end{align*}
The score function of the submodel is defined as
\begin{align*}
  s(D) = \left. \frac{\partial}{\partial t} \log f_t(D) \right|_{t = 0} = \partial_t \log f_t(D).
\end{align*}
The tangent space $\mathcal{T} \subset L_2(D)$
is the closure of the linear span of all score functions
generated by regular parametric submodels in $\mathcal{P}$.

\paragraph{Density Factorization}

Based on the data structure where
$A$ is observed only in the experimental sample
and $(Z, Y)$ are observed only in the observational sample,
the joint density of the observed data $D$ can be factorized as
\begin{align*}
  f(D) &= \bs{f(Z, S \mid X, G = \text{O}) f(Y, W \mid Z, S, X, G = \text{O}}^{\mathbf{1}_{\{G = \text{O}\}}} \\
       &\quad \times \bs{f(A \mid X, G = \text{E}) \cdot f(W, S \mid A, X, G = \text{E}}^{\mathbf{1}_{\{G = \text{E}\}}} \\
       &\quad \times f(X, G).
\end{align*}
Corresponding to this factorization,
the tangent space $\mathcal{T}$
can be decomposed into the direct sum of five orthogonal subspaces:
\begin{align*}
  \mathcal{T} = \Lambda_1 \oplus \Lambda_2 \oplus \Lambda_3 \oplus \Lambda_4 \oplus \Lambda_5.
\end{align*}

\paragraph{Characterization of the Tangent Space}

\begin{enumerate}
  \item According to \cref{eq:outcome-bridge-function},
    for all $t$,
    \begin{align*}
      \E_t \bs{Y - h_t(W, S, X) \mid Z, S, X, G = \text{O}} = 0.
    \end{align*}
    Differentiating the above equation with respect to $t$ at $t = 0$,
    we have
    \begin{align*}
      &\mathrel{\phantom{=}} 
      \E\bs{(Y - h(W, S, X)) s(Y, W \mid Z, S, X, G) \mid Z, S, X, G = \text{O}} \\
      &= \E\bs{\partial_t h_t(W, S, X) \mid Z, S, X, G = \text{O}} \\
      &= T(\partial_t h_t),
    \end{align*}
    where $T$ is the linear operator defined in \cref{sec:estimation-and-inference}.
    Thus, the first subspace $\Lambda_1$ is given by
    \begin{align*}
      \Lambda_1 
      = \left\{
        \mathbf{1}_{\{G = \text{O}\}}s_1(Y, W \mid Z, S, X, G) \;\middle|\;
        \begin{array}{l}
        \E[s_1 \mid Z, S, X, G = \text{O}] = 0, \\
        \E\bs{(Y - h_0) s_1 \mid Z, S, X, G = \text{O}} \in \Range(T)
        \end{array}
      \right\},
    \end{align*}
    where $\Range(T)$ is the range of $T$.

  \item $\Lambda_2$ is the space of score functions
    of the density of proxies and surrogates in the observational sample:
    \begin{align*}
      \Lambda_2 = 
      \left\{
        \mathbf{1}_{\{G = \text{O}\}} s_2(Z, S \mid X, G) \;\middle|\;
        \E[s_2 \mid X, G = \text{O}] = 0
      \right\}.
    \end{align*}

  \item $\Lambda_3$ is the space of score functions
    of the conditional density of proxies and surrogates in the experimental sample:
    \begin{align*}
      \Lambda_3 = 
      \left\{
        \mathbf{1}_{\{G = \text{E}\}} s_3(W, S \mid A, X, G) \;\middle|\;
        \E[s_3 \mid A, X, G = \text{E}] = 0
      \right\}.
    \end{align*}

  \item $\Lambda_4$ is the space of score functions
    of the treatment assignment mechanism
    in the experimental sample, i.e.,
    \begin{align*}
      \Lambda_4 = 
      \left\{
        \mathbf{1}_{\{G = \text{E}\}} s_4(A \mid X, G) \;\middle|\;
        \E[s_4 \mid X, G = \text{E}] = 0
      \right\}.
    \end{align*}

  \item $\Lambda_5$ is the space of score functions
    of the joint density $f(X, G)$, i.e.,
    \begin{align*}
      \Lambda_5 = \big\{s_5(X, G) \mid \E[s_5] = 0 \big\}.
    \end{align*}
\end{enumerate}

\paragraph{Derivation of the Influence Function}

Define
\begin{align*}
  \bar{h}_t(a, X)
  &\equiv \E_t\bs{h_t(W, S, X) \mid A = a, X, G = \text{E}} \\
  &= \iint h_t(w, s, X) f_t(w, s \mid A = a, X, G = \text{E}) \, d w \, d s,
\end{align*}
where $h_t$ and $f_t$ are the outcome bridge function
and the density under $P_t$,
respectively.
By \cref{thm:outcome-bridge-identification},
\begin{align*}
  \tau(P_t) 
  &= \sum_{a \in \{0, 1\}} (-1)^{1 - a} \E_t\bs{\bar{h}_t(a, X) \mid G = \text{E}} \\
  &= \sum_{a \in \{0, 1\}} (-1)^{1 - a} \int \bar{h}_t(a, x) f_t(x \mid \text{E}) \, d x \\
  &= \sum_{a \in \{0, 1\}} (-1)^{1 - a} 
    \int \bp{\iint h_t(w, s, x) f_t(w, s \mid a, x, \text{E}) \, d w \, d s}
    f_t(x \mid \text{E}) \, d x.
\end{align*}
We calculate the pathwise derivative of the target parameter 
$\tau(P_t)$ at $t = 0$:
\begin{equation}
  \begin{split}
    &\mathrel{\phantom{=}} \left. \frac{d}{d t} \tau(P_t) \right|_{t = 0} \\
    &= \sum_{a \in \{0, 1\}} (-1)^{1 - a}
    \bp{\int \bp{\partial_t \bar{h}(a, x)} f_0 (x \mid \text{E}) \, dx
    + \int \bar{h}_0(a, x) \partial_t f_t(x \mid \text{E}) \, dx } \\
    &= \sum_{a \in \{0, 1\}} (-1)^{1 - a}
    \left( \iiint \partial_t h_t(w, s, x) f_0(w, s \mid a, x, \text{E}) f_0(x \mid \text{E}) \, dw \, ds \, dx \right. \\
    &\quad \left. + \iint h_0(w, s, x) \partial_t f_t(w, s \mid a, x, \text{E}) f_0(x \mid \text{E}) \, dw \, ds \, dx 
    + \int \bar{h}_0(a, x) \partial_t f_t(x \mid \text{E}) \, dx \right).
    \label{eq:eif-derivation}
  \end{split}
\end{equation}

We handle each term in \cref{eq:eif-derivation} separately.
Consider the first term involving $\partial_t h_t$.
The inner integral in the first term
can be written as
\begin{align}
  &\mathrel{\phantom{=}} 
  \E_0 \bs{\partial_t h_t(W, S, X) \mid A = a, G = \text{E}} \notag \\
  &\equiv \iiint \partial_t h_t(w, s, x) f_0(w, s \mid a, x, \text{E}) f_0(x \mid \text{E}) \, dw \, ds \, dx \label{eq:eif-derivation-1-1} \\
  &= \iiint \partial_t h_t(w, s, x) \E_0\bs{q_{a, 0}(Z, s, x) \mid w, s, x, \text{O}} f_0(w, s, x \mid \text{O}) \, dw \, ds \, dx \label{eq:eif-derivation-1-2} \\
  &= \E_0 \bs{\partial_t h_t(W, S, X) \E_0\bs{q_{a, 0}(Z, S, X) \mid W, S, X, \text{O}} \mid G = \text{O}} \label{eq:eif-derivation-1-3} \\
  &= \E_0 \bs{\partial_t h_t(W, S, X) q_{a, 0}(Z, S, X) \mid G = \text{O}} \label{eq:eif-derivation-1-4} \\
  &= \E_0 \bs{q_{a, 0}(Z, S, X) \E_0\bs{\partial_t h_t(W, S, X) \mid Z, S, X, G = \text{O}} \mid G = \text{O}}, \label{eq:eif-derivation-1-5}
\end{align}
where 
\cref{eq:eif-derivation-1-1} is by definition,
\cref{eq:eif-derivation-1-2} follows from \cref{eq:surrogate-bridge-function},
\cref{eq:eif-derivation-1-3} is again by definition,
\cref{eq:eif-derivation-1-4,eq:eif-derivation-1-5} use the law of total expectation.
To determine $\E_0\bs{\partial_t h_t(W, S, X) \mid Z, S, X, G = \text{O}}$,
we differentiate \cref{eq:outcome-bridge-function} 
\begin{align*}
  \iint (y - h_t(w, s, x)) f_t(y, w \mid z, s, x, \text{O}) \, dy \, dw = 0
\end{align*}
with respect to $t$ at $t = 0$,
which gives
\begin{align*}
  \iint \partial_t h_t(w, s, x) f_0(y, w \mid z, s, x, \text{O}) \, dy \, dw
  = \iint (y - h_0(w, s, x)) \partial_t f_t(y, w \mid z, s, x, \text{O}) \, dy \, dw.
\end{align*}
Equivalently,
\begin{align}
  &\mathrel{\phantom{=}} \E_0\bs{\partial_t h_t(W, S, X) \mid Z, S, X, G = \text{O}} \\
  &= \E_0\bs{(Y - h_0(W, S, X)) s(Y, W \mid Z, S, X, G) \mid Z, S, X, G = \text{O}} \\
  &= \E_0\bs{(Y - h_0(W, S, X)) (s(Y, W, Z, S, X, G) - s(Z, S, X, G)) \mid Z, S, X, G = \text{O}} \label{eq:eif-derivation-1-6} \\
  &= \E_0\bs{(Y - h_0(W, S, X)) s(Y, W, Z, S, X, G) \mid Z, S, X, G = \text{O}} \label{eq:eif-derivation-1-7} \\
  &= \E_0\bs{(Y - h_0(W, S, X)) s(Y, W, Z, S, A, X, G) \mid Z, S, X, G = \text{O}}, \label{eq:eif-derivation-1-8}
\end{align}
where 
\cref{eq:eif-derivation-1-6} follows from the fact that
\begin{align*}
  s(Y, W \mid Z, S, X, G)
  = s(Y, W, Z, S, X, G) - s(Z, S, X, G),
\end{align*}
\cref{eq:eif-derivation-1-7} holds because $h_0$ satisfies \cref{eq:outcome-bridge-function}, i.e.,
\begin{align*}
  &\mathrel{\phantom{=}} 
  \E_0\bs{(Y - h_0(W, S, X)) s(Z, S, X, G) \mid Z, S, X, G = \text{O}} \\
  &= s(Z, S, X, G) \E_0\bs{Y - h_0(W, S, X) \mid Z, S, X, G = \text{O}} \\
  &= 0,
\end{align*}
and \cref{eq:eif-derivation-1-8} uses the mean-zero property of score functions:
\begin{align*}
  &\mathrel{\phantom{=}} 
  \E_0\bs{(Y - h_0) s(A \mid Y, W, Z, S, X, G) \mid Z, S, X, G = \text{O}} \\
  &= \E_0\bs{(Y - h_0) \underbrace{\E_0\bs{s(A \mid Y, W, Z, S, X, G) \mid Y, W, Z, S, X, G}}_{= 0} \;\middle|\; Z, S, X, G = \text{O}} \\
  &= 0.
\end{align*}
Therefore,
the first term in \cref{eq:eif-derivation} becomes
\begin{align}
  &\mathrel{\phantom{=}} 
  \sum_{a \in \{0, 1\}} (-1)^{1 - a}
  \E_0 \bs{q_{a, 0}(Z, S, X) (Y - h_0(W, S, X)) s(Y, W, Z, S, X, G) \mid G = \text{O}} \notag \\
  &= \E_0 \bs{\bp{q_{1, 0}(Z, S, X) - q_{0, 0}(Z, S, X)} (Y - h_0(W, S, X)) s(Y, W, Z, S, A, X, G) \mid G = \text{O}} \notag \\
  &= \E_0 \bs{\frac{\mathbf{1}_{\{G = \text{O}\}}}{1 - \pi_0} \bp{q_{1, 0}(Z, S, X) - q_{0, 0}(Z, S, X)} (Y - h_0(W, S, X)) s(D)}. \label{eq:eif-derivation-1-9}
\end{align}

Next, we consider the second term in \cref{eq:eif-derivation} 
involving $\partial_t f(w, s \mid a, x, \text{E})$.
Notice that
\begin{align}
  &\mathrel{\phantom{=}} 
  \iint h_0(w, s, x) \partial_t f_t(w, s \mid a, x, \text{E}) f_0(x \mid \text{E}) \, dw \, ds\\
  &= \iint h_0(w, s, x) s(w, s \mid a, x, \text{E}) f_0(w, s \mid a, x, \text{E}) f_0(x \mid \text{E}) \, dw \, ds \label{eq:eif-derivation-2-1} \\
  &= \E_0\bs{h_0(W, S, X) s(W, S \mid A, X, G) \mid A = a, X, G = \text{E}} \label{eq:eif-derivation-2-2} \\
  &= \E_0\bs{\bp{h_0(W, S, X) - \bar{h}_0(a, X)} s(W, S \mid A, X, G) \mid A = a, X, G = \text{E}} \label{eq:eif-derivation-2-3} \\
  &= \E_0\bs{\bp{h_0(W, S, X) - \bar{h}_0(a, X)} s(W, S, A, X, G) \mid A = a, X, G = \text{E}} \label{eq:eif-derivation-2-4} \\
  &= \E_0\bs{\bp{h_0(W, S, X) - \bar{h}_0(a, X)} s(D) \mid A = a, X, G = \text{E}}, \label{eq:eif-derivation-2-5}
\end{align}
where
\cref{eq:eif-derivation-2-1} 
uses the log-derivative property $\partial_t f_t(\cdot) = s(\cdot) f_0(\cdot)$,
\cref{eq:eif-derivation-2-2} is by definition,
\cref{eq:eif-derivation-2-3} holds because
\begin{align*}
  &\mathrel{\phantom{=}} 
  \E_0\bs{\bar{h}_0(a, X) s(W, S \mid A, X, G) \mid A = a, X, G = \text{E}} \\
  &= \bar{h}_0(a, X) \E_0\bs{s(W, S \mid A, X, G) \mid A = a, X, G = \text{E}} \\
  &= 0,
\end{align*}
\cref{eq:eif-derivation-2-4} uses the mean-zero property of score functions,
\begin{align*}
  &\mathrel{\phantom{=}} 
  \E_0\bs{\bp{h_0 - \bar{h}_0} s(A \mid W, S, X, G) \mid A = a, X, G = \text{E}} \\
  &= \E_0\bs{\bp{h_0 - \bar{h}_0} \underbrace{\E_0\bs{s(A \mid W, S, X, G) \mid W, S, X, G}}_{= 0} \;\middle|\; A = a, X, G = \text{E}} \\
  &= 0,
\end{align*}
and \cref{eq:eif-derivation-2-5} follows from the same argument as \cref{eq:eif-derivation-2-4}.
Therefore,
the second term in \cref{eq:eif-derivation} becomes
\begin{align}
  &\mathrel{\phantom{=}} 
  \sum_{a \in \{0, 1\}} (-1)^{1 - a}
  \E_0 \bs{\E_0\bs{\bp{h_0(W, S, X) - \bar{h}_0(a, X)} s(D) \mid A = a, X, G = \text{E}} \mid G = \text{E}} \notag \\
  &= \E_0\bs{\frac{A - e_0(X)}{e_0(X)(1 - e_0(X))} \bp{h_0(W, S, X) - \bar{h}_0(A, X)} s(D) \;\middle|\; G = \text{E}} \label{eq:eif-derivation-2-6} \\
  &= \E_0\bs{\frac{\mathbf{1}_{\{G = \text{E}\}}}{\pi_0} \bp{\frac{A - e_0(X)}{e_0(X)(1 - e_0(X))}} \bp{h_0(W, S, X) - \bar{h}_0(A, X)} s(D)}, \label{eq:eif-derivation-2-7}
\end{align}
where \cref{eq:eif-derivation-2-5} follows from the definition of the propensity score
$e_0(X) \equiv P(A = 1 \mid X, G = \text{E})$,
and \cref{eq:eif-derivation-2-6} is by definition.

Finally, we consider the third term in \cref{eq:eif-derivation}
involving $\partial_t f_t(x \mid \text{E})$:
\begin{align}
  &\mathrel{\phantom{=}} 
  \sum_{a \in \{0, 1\}} (-1)^{1 - a}
  \int \bar{h}_0(a, x) \partial_t f_t(x \mid \text{E}) \, dx \notag \\
  &= \int \bp{\bar{h}_0(1, x) - \bar{h}_0(0, x)} s(x \mid \text{E}) f_0(x \mid \text{E}) \, dx \label{eq:eif-derivation-3-1} \\
  &= \E_0\bs{\bp{\bar{h}_0(1, X) - \bar{h}_0(0, X)} s(X \mid G) \mid G = \text{E}} \label{eq:eif-derivation-3-2} \\
  &= \E_0\bs{\bp{\bar{h}_0(1, X) - \bar{h}_0(0, X)} s(D) \mid G = \text{E}} \label{eq:eif-derivation-3-3} \\
  &= \E_0\bs{\bp{\bar{h}_0(1, X) - \bar{h}_0(0, X) - \tau_0} s(D) \mid G = \text{E}} \label{eq:eif-derivation-3-4} \\
  &= \E_0\bs{\frac{\mathbf{1}_{\{G = \text{E}\}}}{\pi_0} \bp{\bar{h}_0(1, X) - \bar{h}_0(0, X) - \tau_0} s(D)}, \label{eq:eif-derivation-3-5}
\end{align}
where \cref{eq:eif-derivation-3-1} follows from the log-derivative property,
\cref{eq:eif-derivation-3-2} is by definition,
\cref{eq:eif-derivation-3-3,eq:eif-derivation-3-4} use the mean-zero property of score functions,
and \cref{eq:eif-derivation-3-5} is again by definition.

Combining the results from \cref{eq:eif-derivation-1-9,eq:eif-derivation-2-7,eq:eif-derivation-3-5},
we have
\begin{align*}
  \left. \frac{d}{d t} \tau(P_t) \right|_{t = 0}
  &= \E_0\bs{\IF(D) s(D)} \\
  &= \E_0\bs{\bp{\varphi(D; \eta_0) - \frac{\mathbf{1}_{\{G = \text{E}\}}}{\pi_0}\tau_0} s(D)} \\
  &= \E_0\bs{\bp{\varphi_{\text{O}}(D; \eta_0) + \varphi_{\text{E}}(D; \eta_0) - \frac{\mathbf{1}_{\{G = \text{E}\}}}{\pi_0}\tau_0} s(D)},
\end{align*}
where
\begin{align*}
  &\mathrel{\phantom{=}} \IF(D) \\
  &\equiv \varphi(D; \eta_0) - \frac{\mathbf{1}_{\{G = \text{E}\}}}{\pi_0}\tau_0 \\
  &\equiv \varphi_{\text{O}}(D; \eta_0) + \varphi_{\text{E}}(D; \eta_0) - \frac{\mathbf{1}_{\{G = \text{E}\}}}{\pi_0}\tau_0, \\
  &\equiv \frac{\mathbf{1}_{\{G = \text{O}\}}}{1 - \pi_0} \bp{q_{1,0}(Z, S, X) - q_{0,0}(Z, S, X)} (Y - h_0(W, S, X)), \\
  &\quad + \frac{\mathbf{1}_{\{G = \text{E}\}}}{\pi_0} 
  \bp{\frac{\bp{A - e_0(X)}\bp{h_0(W, S, X) - \bar{h}_0(A, X)}}{e_0(X)(1 - e_0(X))} + \bp{\bar{h}_0(1, X) - \bar{h}_0(0, X)} - \tau_0}.
\end{align*}
Thus, we have derived the influence function $\IF(D)$.

\paragraph{Verification of the Tangent Space Membership}

To complete the proof,
we verify that $\IF(D) \in \mathcal{T}$.

First,
the term
$\varphi_{\text{O}}(D; \eta_0)$
is an element of $\Lambda_1$.
To see this,
note that
\begin{align*}
  &\mathrel{\phantom{=}} 
  \E_0\bs{\varphi_{\text{O}}(D; \eta_0) \mid Z, S, X, G = \text{O}} \\
  &= \E_0\bs{\frac{1}{1 - \pi_0} \bp{q_{1, 0}(Z, S, X) - q_{0, 0}(Z, S, X)} (Y - h_0(W, S, X)) \mid Z, S, X, G = \text{O}} \\
  &= \frac{1}{1 - \pi_0} \bp{q_{1, 0}(Z, S, X) - q_{0, 0}(Z, S, X)} \E_0\bs{Y - h_0(W, S, X) \mid Z, S, X, G = \text{O}} \\
  &= 0,
\end{align*}
where the last equality holds because $h_0$ satisfies \cref{eq:outcome-bridge-function}.
Moreover,
\begin{align*}
  &\mathrel{\phantom{=}} 
  \E_0\bs{(Y - h_0(W, S, X)) \varphi_{\text{O}}(D; \eta_0) \mid Z, S, X, G = \text{O}} \\
  &= \E_0\bs{\frac{1}{1 - \pi_0} \bp{q_{1, 0}(Z, S, X) - q_{0, 0}(Z, S, X)} (Y - h_0(W, S, X))^2 \;\middle|\; Z, S, X, G = \text{O}} \\
  &= \frac{1}{1 - \pi_0} \bp{q_{1, 0}(Z, S, X) - q_{0, 0}(Z, S, X)} \E_0\bs{(Y - h_0(W, S, X))^2 \mid Z, S, X, G = \text{O}}.
\end{align*}
By \cref{asn:surjectivity},
$T$ is surjective,
which implies that
the range of $T$ is the entire space $L_2(Z, S, X)$.
Since both $q_{a, 0}(Z, S, X)$ and $\E_0\bs{(Y - h_0(W, S, X))^2 \mid Z, S, X, G = \text{O}}$
are elements of $L_2(Z, S, X)$,
their product is also an element of $L_2(Z, S, X)$.
Thus,
\begin{align*}
  \E_0\bs{(Y - h_0(W, S, X)) \varphi_{\text{O}}(D; \eta_0) \mid Z, S, X, G = \text{O}} \in \Range(T).
\end{align*}

Second,
the term
\begin{align*}
  \frac{\mathbf{1}_{\{G = \text{E}\}}}{\pi_0} \bp{\frac{\bp{A - e_0(X)}\bp{h_0(W, S, X) - \bar{h}_0(A, X)}}{e_0(X)(1 - e_0(X))}}
\end{align*}
is an element of $\Lambda_3$
because
\begin{align*}
  &\mathrel{\phantom{=}} 
  \E_0\bs{\frac{\mathbf{1}_{\{G = \text{E}\}}}{\pi_0} \bp{\frac{\bp{A - e_0(X)}\bp{h_0(W, S, X) - \bar{h}_0(A, X)}}{e_0(X)(1 - e_0(X))}} \;\middle|\; A, X, G = \text{E}} \\
  &= \frac{1}{\pi_0} \bp{\frac{\bp{A - e_0(X)}}{e_0(X)(1 - e_0(X))}} \E_0\bs{h_0(W, S, X) - \bar{h}_0(A, X) \mid A, X, G = \text{E}} \\
  &= 0,
\end{align*}
where the last equality holds due to the definition of $\bar{h}_0(A, X)$.

Finally,
the term
\begin{align*}
  \frac{\mathbf{1}_{\{G = \text{E}\}}}{\pi_0} \bp{\bar{h}_0(1, X) - \bar{h}_0(0, X) - \tau_0}
\end{align*}
is an element of $\Lambda_5$
because
\begin{align*}
  &\mathrel{\phantom{=}}
  \E_0\bs{\frac{\mathbf{1}_{\{G = \text{E}\}}}{\pi_0} \bp{\bar{h}_0(1, X) - \bar{h}_0(0, X) - \tau_0}} \\
  &= \E_0\bs{\bar{h}_0(1, X) - \bar{h}_0(0, X) - \tau_0 \mid G = \text{E}} \\
  &= 0,
\end{align*}
where the last equality holds by the definition of $\tau_0$.

Since each component belongs to the corresponding subspace of the tangent space,
their sum $\IF(D)$ lies in $\mathcal{T}$,
which implies that
$\IF(D)$
is indeed the EIF of the target parameter $\tau_0$.
By the semiparametric efficiency theory \citep[e.g.,][]{bickel1993efficient, tsiatis2006semiparametric},
the semiparametric efficiency bound
is given by the variance of the EIF:
\begin{align*}
  V_{\text{eff}} = \E_0\bs{\IF(D)^2}.
\end{align*}
This completes the proof.

\subsection{Proof of \cref{thm:efficiency}} \label{sec:proof-of-thm-efficiency}

The result follows immediately from \cref{thm:asymptotic-normality} and \cref{thm:eif}.
By \cref{thm:asymptotic-normality}, 
the estimator is asymptotically linear with influence function
$\IF(D)$,
and by \cref{thm:eif},
$\IF(D)$
is the EIF derived in \cref{thm:eif}.

\section{Connection to \cite{athey2025surrogate}'s Identification Results} \label{sec:connection-to-athey2025surrogate}

\subsection{Connection to Surrogate Index Identification} \label{sec:connection-to-surrogate-index-identification}

Suppose there are no unobserved confounders,
that is, $U = \emptyset$,
then \cref{asn:unconfoundedness-observational,asn:unconfoundedness-experimental}
reduce to
\begin{align*}
  A \ind \bp{\bp{Y(a)}_{a \in \{0, 1\}}, \bp{S(a)}_{a \in \{0, 1\}}} &\mid X, G, \\
  S \ind \bp{Y(s)}_{s \in \mathcal{S}} &\mid A, X, G.
\end{align*}
and \cref{asn:availability-of-proxies} reduces to
\begin{align*}
  Z \ind Y \mid S, X, \quad W \ind (A, S, Z) \mid X.
\end{align*}
In this case, we have
\begin{align*}
  &\mathrel{\phantom{=}} \E\bs{Y \mid Z, S, X, G = \text{O}} \\
  &= \E\bs{h_0(W, S, X) \mid S, X, G = \text{O}} \tag{\cref{asn:outcome-bridge-function}} \\
  &= \E\bs{Y \mid S, X, G = \text{O}} \tag{$Z \ind Y \mid (S, X)$} \\
  &= \E\bs{\E\bs{Y \mid W, S, X, G = \text{O}} \mid S, X, G = \text{O}}, \tag{Iterated expectation}
\end{align*}
By \cref{asn:completeness-outcome},
comparing the the first and last equations above,
a solution $h_0$ to the above equation is
the surrogate index in \cite{athey2025surrogate}:
\begin{align*}
  h_0(w, s, x) = \E\bs{Y \mid S = s, X = x, G = \text{O}} \equiv \mu(s, x, \text{O}).
\end{align*}
Since $h_0(w, s, x)$ does not depend on $w$ in this case,
the identification result in \cref{eq:outcome-bridge-identification}
reduces to
\begin{align*}
  \tau_0
  &= \E\bs{\frac{A h_0(W, S, X)}{e_0(X)} - \frac{(1 - A) h_0(W, S, X)}{1 - e_0(X)} \;\middle|\; G = \text{E}} \\
  &= \E\bs{\frac{A \mu(S, X, \text{O})}{e_0(X)} - \frac{(1 - A) \mu(S, X, \text{O})}{1 - e_0(X)} \;\middle|\; G = \text{E}}.
\end{align*}
This matches the surrogate index representation in \cite{athey2025surrogate}.

\subsection{Connection to Surrogate Score Identification} \label{sec:connection-to-surrogate-score-identification}

Similarly,
when there are no unobserved confounders,
we have
\begin{align*}
  &\mathrel{\phantom{=}} \E\bs{q_{a, 0}(Z, S, X) \mid W, S, X, G = \text{O}} \\
  &= \frac{f(W, S \mid A = a, X, G = \text{E}) f(X \mid G = \text{E})}{f(W, S, X \mid G = \text{O})} \tag{\cref{asn:surrogate-bridge-function}} \\
  &= \frac{f(W \mid S, A = a, X, G = \text{E}) f(S \mid A = a, X, G = \text{E}) f(X \mid G = \text{E})}{f(W, S, X \mid G = \text{O})} \tag{Chain rule} \\
  &= \frac{f(W \mid X, G = \text{E}) f(S \mid A = a, X, G = \text{E}) f(X \mid G = \text{E})}{f(W, S, X \mid G = \text{O})} \tag{$W \ind (A, S) \mid (X, G)$} \\
  &= \frac{f(W \mid X, G = \text{E}) f(S \mid A = a, X, G = \text{E}) f(X \mid G = \text{E})}{f(W \mid S, X, G = \text{O}) f(S, X \mid G = \text{O})} \tag{Chain rule} \\
  &= \frac{f(W \mid X, G = \text{E}) f(S \mid A = a, X, G = \text{E}) f(X \mid G = \text{E})}{f(W \mid X, G = \text{O}) f(S, X \mid G = \text{O})} \tag{$W \ind S \mid (X, G)$} \\
  &= \frac{f(W \mid X, G = \text{O}) f(S \mid A = a, X, G = \text{E}) f(X \mid G = \text{E})}{f(W \mid X, G = \text{O}) f(S, X \mid G = \text{O})} \tag{$W \ind G \mid X$} \\
  &= \frac{f(S \mid A = a, X, G = \text{E}) f(X \mid G = \text{E})}{f(S, X \mid G = \text{O})}. \tag{Simplification}
\end{align*}
It is clear that the right-hand side above does not depend on $Z$ and $W$,
so we have a solution $q_a^*$ to \cref{eq:surrogate-bridge-function}:
\begin{align*}
  q_{a, 0}(S, X)
  = \frac{f(S\mid A = a, X, G = \text{E}) f(X \mid G = \text{E})}{f(S, X \mid G = \text{O})},
\end{align*}
which only depends on $S$ and $X$.

Note that
\begin{align*}
  &\mathrel{\phantom{=}}
  \frac{f(S \mid A = a, X, G = \text{E}) f(X \mid G = \text{E})}{f(S, X \mid G = \text{O})} \\
  &= \frac{f(A = a \mid S, X, G = \text{E}) f(S \mid X G = \text{E}) f(X \mid G = \text{E})}{f(A = a \mid X, G = \text{E}) f(S, X \mid G = \text{O})} \tag{Bayes' rule} \\
  &= \frac{f(A = a \mid S, X, G = \text{E}) f(S, X \mid G = \text{E})}{f(A = a \mid X, G = \text{E}) f(S, X \mid G = \text{O})} \tag{Chain rule} \\
  &= \frac{f(A = a \mid S, X, G = \text{E})}{f(A = a \mid X, G = \text{E})}
  \cdot
  \frac{f(S, X \mid G = \text{E})}{f(S, X \mid G = \text{O})} \tag{Separating fractions} \\
  &= \frac{f(A = a \mid S, X, G = \text{E})}{f(A = a \mid X, G = \text{E})}
  \cdot
  \frac{f(G = \text{E} \mid S, X)}{f(G = \text{O} \mid S, X)} 
  \cdot
  \frac{f(G = \text{O})}{f(G = \text{E})}, \tag{Bayes' rule}
\end{align*}
which is exactly the weight of the surrogate score representation in \cite{athey2025surrogate}.

\subsection{Connection to Influence Function Identification} \label{sec:connection-to-influence-function-identification}

When there are no unobserved confounders,
we have shown in \cref{sec:connection-to-surrogate-index-identification,sec:connection-to-surrogate-score-identification}
that $h^*$ and $q_a^*$ are solutions to \cref{eq:outcome-bridge-function,eq:surrogate-bridge-function} respectively,
with
\begin{align*}
  h^*(W, S, X) &= \mu(S, X, \text{O}), \\
  \bar{h}^*(A, X) &= \bar{\mu}(A, X, \text{O}) = \E\bs{\mu(S, X, \text{O}) \mid A, X, G = \text{E}}, \\
  q_a^*(Z, S, X) &= \frac{f(A = a \mid S, X, G = \text{E})}{f(A = a \mid X, G = \text{E})}
  \cdot
  \frac{f(G = \text{E} \mid S, X)}{f(G = \text{O} \mid S, X)} 
  \cdot
  \frac{f(G = \text{O})}{f(G = \text{E})}.
\end{align*}

Substituting them into \cref{eq:mutiply-robust-identification},
it is clear that the first two terms coincide with the first two terms
in the influence function derived in \cite{athey2025surrogate}.
As for the third term,
we have
\begin{align*}
  &\mathrel{\phantom{=}}
  q_1^*(Z, S, X) - q_0^*(Z, S, X) \\
  &= \frac{f(G = \text{O})}{f(G = \text{E})} \cdot
  \frac{f(G = \text{E} \mid S, X)}{f(G = \text{O} \mid S, X)}
  \cdot
  \bs{
    \frac{f(A = 1 \mid S, X, G = \text{E})}{f(A = 1 \mid X, G = \text{E})}
    -
    \frac{f(A = 0 \mid S, X, G = \text{E})}{f(A = 0 \mid X, G = \text{E})}
  } \\
  &= \frac{f(G = \text{O})}{f(G = \text{E})} \cdot
  \frac{f(G = \text{E} \mid S, X)}{f(G = \text{O} \mid S, X)} \\
  &\quad \cdot
  \frac{f(A = 1 \mid S, X, G = \text{E})(1 - e(X)) - (1 - f(A = 1 \mid S, X, G = \text{E}))e(X)}{e(X)(1 - e(X))} \\
  &= \frac{f(G = \text{O})}{f(G = \text{E})} \cdot
  \frac{f(G = \text{E} \mid S, X)}{f(G = \text{O} \mid S, X)}
  \cdot
  \frac{f(A = 1 \mid S, X, G = \text{E}) - e(X)}{e(X)(1 - e(X))},
\end{align*}
which matches the weight of the third term
in the influence function derived in \cite{athey2025surrogate}.

\end{document}